\documentclass[journal]{IEEEtran}

\usepackage{subcaption}
\usepackage{amsmath,amsfonts}
\usepackage{algorithm}
\usepackage{algpseudocode}
\usepackage{array}
\usepackage{textcomp}
\usepackage{stfloats}
\usepackage{url}
\usepackage{verbatim}
\usepackage{graphicx}
\usepackage{xcolor}
\usepackage{booktabs}
\usepackage{amsthm}

\newtheorem{theorem}{Theorem}
\newtheorem{myDef}{Definition}

\newtheorem{corollary}{Corollary}
\newtheorem{lemma}{Lemma}
\usepackage{tabularx}
\usepackage{amsthm} 
\usepackage{mathtools}

\ifCLASSINFOpdf

\else

\fi

\begin{document}

\title{PAC-DP: Personalized Adaptive Clipping for Differentially Private Federated Learning}

\author{ 
Hao~Zhou,
Siqi~Cai,
Hua~Dai*,
Geng~Yang,
Jing~Luo,
Hui~Cai
\thanks{
H. Zhou is with the school of computer science, Nanjing University of Post and Telecommunication,
the State Key Laboratory of Tibetan Intelligence,
and the State Key Laboratory for Novel Software Technology,Nanjing University, Nanjing, P.R. China.
E-mail: haozhou@njupt.edu.cn.
}
\thanks{
S. Cai, H. Dai, G. Yang, J. Luo, H. Cai are are with the school of computer science, Nanjing University of Post and Telecommunication, Nanjing 210023, China.
E-mail: \{1024041101, daihua, yangg, 1024041121, carolinecai\}@njupt.edu.cn.
}
\thanks{
This work is supported by the National Natural Science Foundation of China (62372244, 62572253), 
Jiangsu Provincial Natural Science Foundation (Youth Fund) (BK20250683),
Natural Science Foundation of Nanjing University of Posts and Telecommunications(NY224058, NY225124),
Young Elite Scientists Sponsorship Program by CAST (JSTJ-2025-641),
Basic Science (Natural Science) Research Project of Higher Education Institutions in Jiangsu Province (25KJB520040),
Open Project of State Key Laboratory for Novel Software Technology (KFKT2025B68),
and China Postdoctoral Special Grant Foundation (2024T170432).
}
}

\maketitle

\begin{abstract}
Differential privacy (DP) is crucial for safeguarding sensitive client information in federated learning (FL), yet traditional DP-FL methods rely predominantly on fixed gradient clipping thresholds. 
Such static clipping neglects significant client heterogeneity and varying privacy sensitivities, which may lead to an unfavorable privacy--utility trade-off.
In this paper, we propose PAC-DP, a Personalized Adaptive Clipping framework for federated learning under \emph{record-level local differential privacy}. 
PAC-DP introduces a Simulation-CurveFitting approach leveraging a server-hosted public proxy dataset to learn an effective mapping between personalized privacy budgets $\varepsilon$ and gradient clipping thresholds $C$, which is then deployed online with a lightweight round-wise schedule.
This design enables budget-conditioned threshold selection while avoiding data-dependent tuning during training.
We provide theoretical analyses establishing convergence guarantees under the per-example clipping and Gaussian perturbation mechanism and a reproducible privacy accounting procedure.
Extensive evaluations on multiple FL benchmarks show that PAC-DP surpasses conventional fixed-threshold approaches under matched privacy budgets, improving accuracy by up to 26\% and accelerating convergence by up to 45.5\% in our evaluated settings.
\end{abstract}

\IEEEpeerreviewmaketitle

\section{Introduction}

The proliferation of mobile and IoT devices in the big data era has led to the generation of vast volumes of decentralized and heterogeneous data \cite{DBLP:journals/tkde/FanGCCCCFGGLLORSSTTTWWX25, chen2023eefl, lee2019deep, li2020review}. Individually, these data sources often lack the volume or variety required for effective model training. At the same time, traditional data collection methods have become infeasible due to privacy regulations, commercial constraints, and growing public concern over data misuse \cite{ma2021communication, 9743558, yao2024ferrari}.

\begin{figure}[!t]  
    \centering  
    \includegraphics[width=\columnwidth, trim=11.5cm 12cm 10cm 5cm, clip]{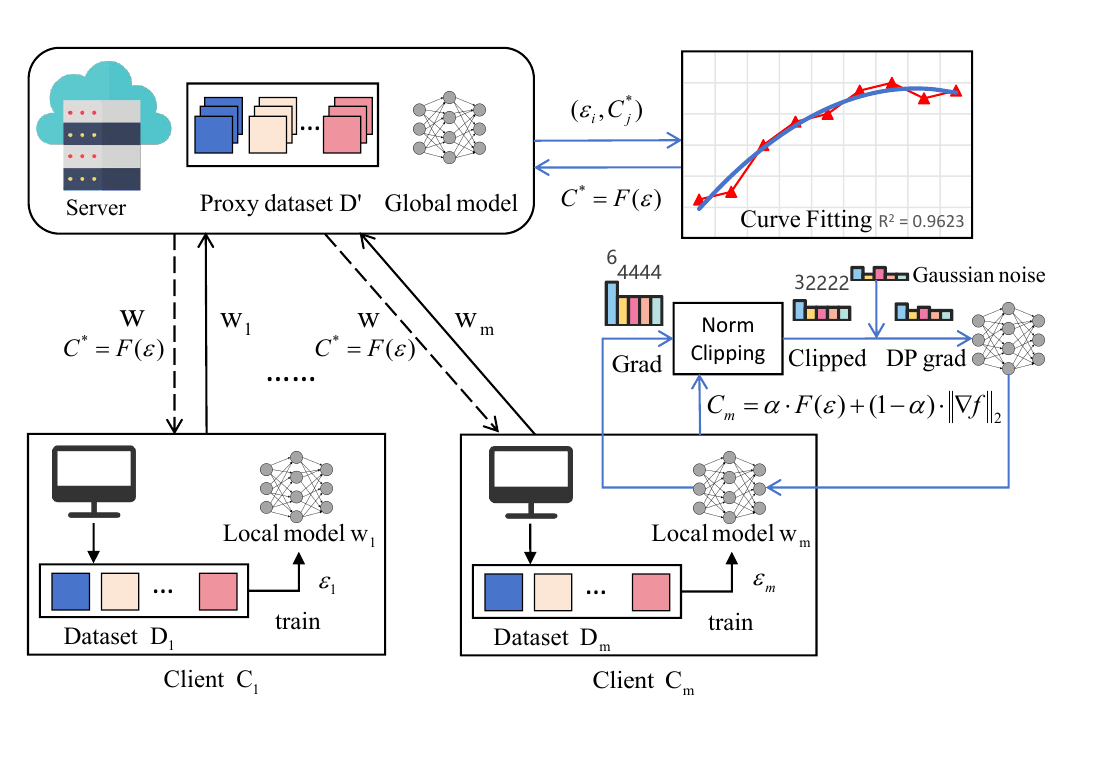}
    \caption{Personalized Adaptive Clipping Mechanism}  
    \label{fig:framework_0} 
\end{figure}

Federated Learning (FL) has emerged as a promising distributed machine learning paradigm that enables collaborative model training without exposing raw user data \cite{DBLP:journals/chinaf/YangXYPYM25, DBLP:conf/aaai/WangXWCQS23}. Clients train models locally and transmit model updates (e.g., gradients) to a central server for aggregation. FL has found successful applications in diverse domains such as healthcare, finance, agriculture, and transportation systems \cite{DBLP:journals/twc/WangXYQ24, mammen2021opportunities, deng2020adaptive, wu2020fedhome}. However, despite avoiding raw data sharing, FL remains susceptible to privacy attacks: model updates can still leak sensitive information \cite{wei2020federated, mcmahan2017learning}, including user participation or partial data reconstruction \cite{wang2022quantized}. Additionally, FL incurs high communication overhead and depends on the stability and trustworthiness of participating devices \cite{konecny2016federated, mcmahan2016communication}.

Federated learning has adopted differential privacy (DP) as its de-facto compliance layer. By perturbing model updates on-device before they ever traverse a network boundary, DP provides formal guarantees that limit information leakage from uploaded messages, which can mitigate membership and reconstruction risks under standard threat models \cite{dwork2013algorithmic, fu2024differentially, fu2024dpsur, xiang2023practical}. Recent advances introduced personalized DP frameworks and mechanisms, allowing heterogeneous clients to specify varying privacy levels while maintaining acceptable utility \cite{yang2024adadp, fu2022adap, deng2020adaptive, collins2021exploiting, shamsian2021personalized, wu2023personalized, zhang2022personalized, ye2023personalized, zhou2024personalized, yang2023personalized, chen2023efficient}. Yet, practical integration of DP into FL systems faces critical challenges, especially regarding gradient clipping for a crucial preprocessing step influencing noise calibration and learning performance.

Most DP-FL methods apply fixed gradient clipping thresholds, neglecting non-IID data distributions, varying gradient magnitudes, and personalized privacy preferences \cite{yang2024adadp, fu2022adap}. Existing adaptive clipping methods typically rely on heuristic strategies (e.g., historical gradient percentiles) lacking theoretical rigor, potentially degrading model performance or even exposing privacy vulnerabilities when gradient distributions change rapidly \cite{fu2022adap, xu2020adadp, wang2022squeezing}.

This gap raises a critical design question: given target privacy loss and evolving non-IID gradient distributions, how can we dynamically determine the optimal clipping bound to maximize utility while respecting diverse client privacy budgets?

To tackle these issues, we introduce \textbf{PAC-DP}, a Personalized Adaptive Clipping framework that enhances differential privacy in federated learning. PAC-DP uniquely employs a Simulation-CurveFitting strategy using a server-side proxy dataset, dynamically establishing mappings between personalized privacy budgets $\varepsilon$ and optimal adaptive gradient clipping thresholds $C^*$.
Departing from existing heuristic-based approaches (e.g., rPDP-FL's percentile clipping \cite{liu2022privacy} or NbAFL's fixed global bound \cite{wei2020federated}) PAC-DP adopts a data-driven optimization strategy by constructing an empirical mapping F($\epsilon$) between user-specified privacy budgets and optimal clipping thresholds through systematic offline simulation on a proxy dataset and assigns budget-conditioned clipping thresholds across heterogeneous clients via the learned mapping $F(\varepsilon)$ and a shared schedule $\lambda(t)$, without using any client-data-dependent threshold selection.

Beyond algorithmic design, PAC-DP also advances the theoretical understanding of personalized DP in federated settings. Unlike existing methods that either lack convergence analysis or rely on fixed clipping, PAC-DP offers rigorous convergence guarantees under \textit{non-convex}, \textit{convex}, and \textit{strongly convex} settings, while explicitly modeling the interaction between clipping thresholds, privacy budgets, and utility degradation.
More importantly, PAC-DP theoretically explains why PAC-DP’s simulation-fitted function $F(\varepsilon)$ achieves better performance than heuristic or fixed-threshold approaches. Such privacy-budget-aware convergence analysis is absent in prior works and offers new insights into privacy-utility trade-offs in DP-FL.

Our primary contributions include:
\begin{itemize}
\item Proposing PAC-DP, an adaptive gradient clipping method enabling personalized differential privacy tailored explicitly for heterogeneous federated learning scenarios.
\item Providing convergence guarantees and a reproducible privacy accounting procedure under record-level local DP with practical client participation.
\item Extensive empirical validation demonstrating PAC-DP’s superior performance compared to fixed-threshold approaches, achieving up to 26\% accuracy improvement, 45.5\% faster convergence, and enhanced privacy-utility trade-offs.
\end{itemize}

\section{Related Work}

\textbf{Differential Privacy in Federated Learning.}
Differential Privacy (DP) has become a foundational technique for safeguarding user data in Federated Learning (FL). It provides formal guarantees that the inclusion or exclusion of a single individual's data does not significantly affect the output of an algorithm, thereby mitigating risks of inference or membership attacks \cite{wang2022quantized, wang2023batch, zhang2024flpurifier, zhang2024badcleaner}. In FL, DP mechanisms are typically applied by injecting noise into model updates or gradients before aggregation, with early methods such as DP-SGD \cite{abadi2016deep} demonstrating the feasibility of privacy-preserving optimization via Gaussian perturbations.
Subsequent research has extended DP to the federated setting in more complex scenarios. For instance, NbAFL \cite{wei2020federated} applies local noise addition at each client before model aggregation, enabling client-level DP and analyzing the trade-off between privacy guarantees and convergence behavior. The Dordis protocol \cite{jiang2024dordis} further optimizes DP-FL execution under unreliable or dropping clients, reducing synchronization overhead and improving protocol robustness. Recent efforts have also integrated DP with backdoor defense strategies to simultaneously mitigate privacy and security threats in FL \cite{zhang2024flpurifier, zhang2024badcleaner}.
A critical challenge in this area is the privacy-utility trade-off, particularly under non-IID data distributions, which are common in real-world FL deployments. Uniform noise addition can disproportionately affect clients with rare or skewed data, leading to substantial utility degradation \cite{wei2020federated}. Techniques such as cross-silo FL with personalized DP \cite{liu2024crosssilo} aim to address this by assigning distinct privacy budgets to each data record or client, enhancing flexibility but complicating privacy accounting and system design.

\textbf{Gradient Clipping in Differentially Private Optimization.}
Unfortunately, the one‑size‑fits‑all clipping rule inherited from DP‑SGD ignores two realities of federated ecosystems: (i) gradient norms vary by task domain (e.g., radiology vs. keyboard prediction) and by client hardware, and (ii) privacy budgets now differ per tenant under multi‑silo FL \cite{abadi2016deep, geyer2017differentially}.
To mitigate the limitations of fixed clipping, several adaptive gradient clipping techniques have been proposed. Adaptive Quantile Clipping (AQC) \cite{geyer2017differentially} estimates a target quantile of gradient norms to dynamically adjust the clipping threshold. This approach reduces unnecessary information loss and improves utility. Similarly, Andrew et al. \cite{andrew2019differentially} propose tracking gradient norms over multiple rounds to tune clipping levels in response to training dynamics.
In federated settings, adaptive clipping methods are increasingly integrated with DP mechanisms to handle gradient heterogeneity and non-IID data. Fu et al. \cite{fu2024dpsur} and Liu et al. \cite{liu2024crosssilo} propose heuristics that adjust clipping thresholds using historical gradient statistics, which can help balance privacy-utility trade-offs. However, these methods often rely on empirical hyperparameters such as fixed percentiles and historical window sizes, which may be difficult to tune and fail to adapt in dynamic environments.
Recent works such as DP-FedMeta \cite{yang2024adadp} and SqueezingFL \cite{wang2022squeezing} incorporate adaptive clipping into federated meta-learning frameworks. These methods use past differentially private gradients to estimate new clipping thresholds while attempting to maintain privacy guarantees. Yet, a common limitation across existing approaches is the lack of principled analysis connecting the clipping threshold to the privacy budget \( \varepsilon \) or learning objective. Moreover, reliance on historical gradients may lead to stale estimates that misalign with heterogeneous privacy budgets and round-wise training dynamics, especially in systems with high client variability.

\textbf{Positioning of This Work.}
In contrast to prior work, our method PAC-DP aims to bridge the gap between theory and practice by explicitly modeling the relationship between the privacy budget \( \varepsilon \) and the optimal clipping threshold \( C^* \). Instead of relying on heuristic tuning or fixed historical statistics, PAC-DP leverages a Simulation-CurveFitting approach to learn the \( \varepsilon \mapsto C^* \) mapping from a server-side proxy dataset. This enables principled, real-time adaptation of clipping thresholds per training round, allowing clients to receive clipping levels consistent with their personalized DP budgets and heterogeneous privacy budgets and round-wise training dynamics. To our knowledge, this is the first work to provide both theoretical analysis and practical mechanisms for privacy-budget-aligned adaptive clipping in federated learning.

\section{Preliminaries}

In federated learning each client \( i \in \{1, \ldots, N\} \) holds local dataset \( \mathcal{D}_i \), and model updates are computed locally and aggregated by a central server without exchanging raw data \cite{mcmahan2017communication}. The global objective is:
\[
w^* = \arg \min_w \sum_{i=1}^N p_i F_i(w), \quad p_i = \frac{|\mathcal{D}_i|}{\sum_j |\mathcal{D}_j|},
\]
where \( F_i(w) \) is the empirical loss on client \( i \). The server aggregates updates as:
\[
w = \sum_{i=1}^N p_i w_i.
\]

This process iterates over local training, server aggregation, and global model broadcasting until convergence.

\begin{myDef}[$(\varepsilon, \delta)$-DP]
A randomized mechanism \( \mathcal{M} \) satisfies $(\varepsilon, \delta)$-DP if for any two adjacent datasets \( D, D' \) and any measurable set \( S \) \cite{dwork2013algorithmic}:
\[
\Pr[\mathcal{M}(D) \in S] \leq e^\varepsilon \Pr[\mathcal{M}(D') \in S] + \delta.
\]
\end{myDef}

The Gaussian mechanism adds $\xi\sim\mathcal N(0,(z\,\Delta s)^2 I)$, where $z$ is the noise multiplier.
To achieve $(\varepsilon,\delta)$-DP, it suffices to set
\[
z \ge \frac{\sqrt{2\ln(1.25/\delta)}}{\varepsilon}.
\]

In FL, DP is typically enforced via local noise addition to gradients or model updates.

\begin{myDef}[Record-level Local DP]
For client $i$ in round $t$, let $B_i^t$ be the local minibatch used to compute the update. 
Two minibatches $B_i^t$ and ${B_i^t}'$ are adjacent if they differ in exactly one record.
A client-side mechanism $\mathcal M_i^t$ satisfies $(\varepsilon_i^t,\delta)$-local DP if for all adjacent $B_i^t,{B_i^t}'$ and all measurable sets $S$ \cite{dwork2013algorithmic}:
\[
\Pr[\mathcal M_i^t(B_i^t)\in S]\le e^{\varepsilon_i^t}\Pr[\mathcal M_i^t({B_i^t}')\in S]+\delta.
\]
\end{myDef}

\textbf{Threat Model}
We consider an honest-but-curious server that observes the participation set $\mathcal S_t$ and all uploaded messages (perturbed gradients or model updates) from participating clients, as well as the broadcast global model in each round. The adversary does not observe any raw client data, per-example gradients, or intermediate unclipped values. Our privacy guarantee is defined with respect to the server’s observable view and does not rely on secure channels. We do not provide privacy guarantees for client participation (i.e., whether a client participates in a given round); our guarantees focus on record-level privacy within the uploaded messages.

\section{Problem Formulation}

We consider a federated learning system consisting of a central server and \( N \) clients, where each client \( i \in \{1, \ldots, N\} \) holds a private local dataset \( \mathcal{D}_i \). The global objective is to learn a model parameter vector \( w \in \mathbb{R}^d \) by minimizing a weighted empirical loss:
\begin{equation}
\min_{w \in \mathbb{R}^d} F(w) := \sum_{i=1}^{N} p_i F_i(w), \quad \text{where } p_i = \frac{|\mathcal{D}_i|}{\sum_j |\mathcal{D}_j|}
\label{eq:global_obj}
\end{equation}
where \( F_i(w) \) denotes the local loss function on client \( i \).

Each client computes a stochastic gradient \( g_i^t \) of the loss with respect to the local model at training round \( t \). To enforce $(\varepsilon_i, \delta)$-local differential privacy (LDP) for each client, noise must be added to the gradients prior to upload. However, the magnitude of the required noise is tied to the sensitivity of the gradient, which is controlled by gradient clipping.

\textbf{Record-level Local DP via Per-example Clipping.}
In round $t$, client $i$ samples a local minibatch $B_i^t$ of size $B$ and computes per-example gradients $g_i^t(x)$ for $x\in B_i^t$.
It applies per-example $\ell_2$ clipping:
\begin{equation}
\tilde g_i^t(x)= g_i^t(x)\cdot \min\left(1, \frac{C_i^t}{\|g_i^t(x)\|_2}\right).
\label{eq:clip_per_example}
\end{equation}
Then it forms the averaged clipped gradient:
\begin{equation}
\bar g_i^t=\frac{1}{B}\sum_{x\in B_i^t}\tilde g_i^t(x),
\label{eq:avg_clipped}
\end{equation}
and adds Gaussian noise using a \emph{noise multiplier} $z_i^t$:
\begin{equation}
\hat g_i^t=\bar g_i^t+\xi_i^t,\quad \xi_i^t\sim \mathcal N\!\left(0,\left(z_i^t\frac{C_i^t}{B}\right)^2 I\right).
\label{eq:dp_noise_multiplier}
\end{equation}

This mechanism satisfies \((\varepsilon_i, \delta)\)-DP for each client. However, selecting an inappropriate \( C_i^t \) can lead to large \( C_i^t \) with higher noise (weakened utility) and small \( C_i^t \) with excessive clipping (information loss).

\textbf{Adaptive Clipping for Personalized DP.}
Our central goal is to design an adaptive clipping mechanism that dynamically selects the optimal clipping threshold \( C_i^t \) per client and round, guided by the client’s privacy budget \( \varepsilon_i \) and heterogeneous privacy budgets and round-wise training dynamics.

We define the optimal clipping threshold \( C^*(\varepsilon_i) \) for client \( i \) as the value minimizing expected utility loss under a fixed privacy budget:
\begin{equation}
C^*(\varepsilon_i) = \arg \min_{C > 0} \mathbb{E}\left[\mathcal{L}(C, \varepsilon_i)\right],
\label{eq:optimal_clip}
\end{equation}
where \( \mathcal{L}(C, \varepsilon_i) \) captures the degradation in model utility due to the combined effect of clipping and noise at threshold \( C \).

In practice, directly solving Eq.~\eqref{eq:optimal_clip} is challenging, as \( \mathcal{L}(\cdot) \) is not available in closed form and depends on task-specific distributions. We therefore seek to approximate \( C^*(\varepsilon_i) \) via empirical simulation and curve fitting over a proxy dataset—a strategy formalized in our proposed method, PAC-DP.
\section{Federated Learning with Adaptive Gradient Clipping Personalized DP}

We consider a standard federated learning setting with a central server and \( N \) edge clients. Each client \( i \in \{1, \dots, N\} \) holds a private dataset \( \mathcal{D}_i \) and participates in collaboratively learning a global model \( w \in \mathbb{R}^d \). The objective is to minimize a weighted empirical loss function:
\begin{equation}
\min_{w \in \mathbb{R}^d} F(w) := \sum_{i=1}^{N} p_i F_i(w), \quad \text{where } p_i = \frac{|\mathcal{D}_i|}{\sum_{j=1}^N |\mathcal{D}_j|}.
\label{eq:global_objective}
\end{equation}
Here, \( F_i(w) := \mathbb{E}_{z \sim \mathcal{D}_i}[\ell(w; z)] \) denotes the expected loss on client \( i \).

To protect user data, each client must satisfy \((\varepsilon_i, \delta)\)-local differential privacy. The gradient update is therefore perturbed locally before transmission:
\begin{equation}
\hat g_i^t = \bar g_i^t + \xi_i^t,\quad 
\xi_i^t\sim \mathcal N\!\left(0,\left(z_i^t\frac{C_i^t}{B}\right)^2 I\right),
\end{equation}
where $\bar g_i^t=\frac{1}{B}\sum_{x\in B_i^t}\mathrm{clip}(g_i^t(x),C_i^t)$ is the minibatch average of per-example clipped gradients. The choice of \( C_i^t \) critically impacts both the amount of noise added and the amount of information retained from the gradient, thus influencing the privacy-utility trade-off.

\subsection{Challenges in Personalized DP Clipping}

The design of an effective personalized DP mechanism for FL presents several key challenges:

\begin{itemize}
    \item \textbf{Heterogeneous Gradients}: Due to non-IID data across clients, gradient magnitudes vary significantly, rendering fixed clipping thresholds suboptimal.
    \item \textbf{Personalized Privacy Budgets}: Clients may select different \( \varepsilon_i \) values based on sensitivity or regulatory compliance, requiring clipping thresholds to be tailored accordingly.
    \item \textbf{Efficiency and Scalability}: Clipping threshold selection must be computationally efficient and avoid real-time hyperparameter tuning or storage of historical gradients.
\end{itemize}

Existing approaches often rely on static heuristics (e.g., fixed percentiles) or historical statistics, which are brittle in dynamic or heterogeneous environments and lack theoretical foundations.

\begin{figure*}[!t]  
    \centering  
    \includegraphics[width=\textwidth, trim=13cm 9.5cm 9.5cm 10cm, clip]{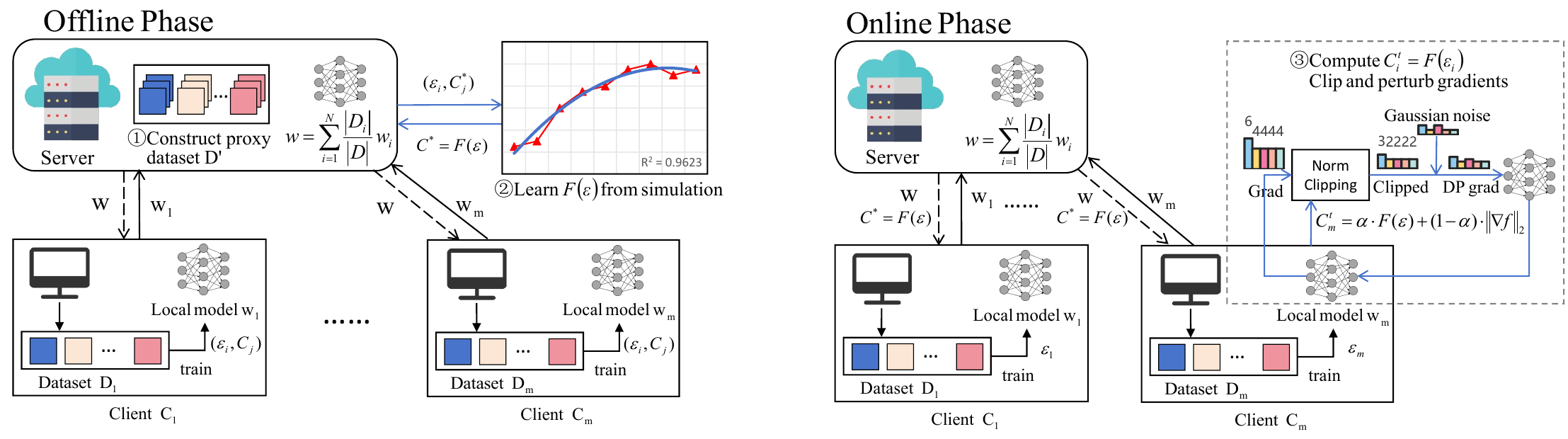}
    \caption{Overview of the PAC-DP framework. A proxy dataset is used to simulate privacy-utility trade-offs and fit a mapping \( C^* = F(\varepsilon) \), which guides personalized clipping and noise injection during training.}  
    \label{fig:framework_1} 
\end{figure*}

\subsection{Solution Overview}

To address these challenges, we propose \textbf{PAC-DP}, a federated optimization framework that enables \textit{adaptive gradient clipping under personalized DP constraints}. The key idea is to learn a smooth, parameterized function \( F(\varepsilon) \) that maps each privacy budget \( \varepsilon \) to its corresponding optimal clipping threshold \( C^* \), i.e.,
\[
C^* = F(\varepsilon).
\]

This mapping is computed \textit{offline} using a server-side proxy dataset through a simulation-curve fitting process, and then used \textit{online} during training to assign each client its personalized clipping threshold.

An overview of the PAC-DP workflow is shown in Fig.~\ref{fig:framework_1}.
Specifically, the PAC-DP framework consists of two core phases:
\begin{itemize}
    \item \textbf{Offline Phase:} The server constructs a proxy dataset \( \hat{D} \) that reflects the global data distribution, simulates training under various \((\varepsilon, C)\) configurations, and fits a function \( F(\varepsilon) \) using regression.
    \item \textbf{Online Phase:} During each round of FL, clients compute \( C_i^t = F(\varepsilon_i) \) locally and perform DP-SGD with clipping and noise injection. The server aggregates the perturbed updates to refine the global model.
\end{itemize}

\subsection{Round-wise Schedule for Clipping}
\label{subsec:schedule}

PAC-DP sets the per-round clipping bound as
\begin{equation}
C_i^t = F(\varepsilon_i)\cdot \lambda(t),
\label{eq:Ci_schedule}
\end{equation}
where $F(\varepsilon_i)$ is the offline learned budget-conditioned mapping and $\lambda(t)$ is a lightweight
round-wise schedule shared across clients.

\paragraph{Schedule design.}
Let $T$ be the total number of communication rounds and $t\in\{0,1,\ldots,T-1\}$.
We use a plateau-then-decay schedule:
\begin{equation}
\lambda(t)=
\begin{cases}
1, & 0\le t < T_s,\\[4pt]
\lambda_{\min} + (1-\lambda_{\min})\cdot \frac{1+\cos\!\left(\pi\cdot \frac{t-T_s}{T-T_s}\right)}{2},
& T_s \le t \le T-1,
\end{cases}
\label{eq:lambda_cosine}
\end{equation}
where $T_s=\lfloor r_s T\rfloor$, $r_s\in(0,1)$ controls when the decay starts, and
$\lambda_{\min}\in(0,1]$ is the final minimum scaling factor.
By construction, $\lambda(t)\in[\lambda_{\min},1]$ and is non-increasing.

Unless otherwise stated, we fix $(r_s,\lambda_{\min})=(0.6,0.1)$ for all datasets and model architectures.
This avoids data-dependent tuning and ensures reproducibility.

Early in training, gradients typically have larger norms and an aggressive decay may induce excessive clipping bias;
hence we keep $\lambda(t)=1$ for the first $r_sT$ rounds.
In later stages, decreasing $\lambda(t)$ reduces the DP noise scale (which is proportional to $C_i^t$ in our mechanism),
improving stability and final utility.

\subsection{Adaptive Personalized DP Training Algorithm}

We formalize the PAC-DP training procedure in the following algorithm:

\begin{algorithm}[htbp]
\caption{Adaptive Gradient Clipping Personalized DP Federated Learning (PAC-DP)}
\label{alg:PAC-DP}
\begin{algorithmic}[1]
\Require Total rounds \( T \), initial model \( w^{(0)} \), privacy budgets \( \{\varepsilon_i\}_{i=1}^N \), proxy dataset \( \hat{D} \)
\Ensure Final model \( w^{(T)} \)
\State \textbf{[Offline Phase]}
\State Construct proxy dataset $\hat{D}$ from public data
\State Learn \( F(\varepsilon) \) from simulation (see Algorithm 2)
\State \textbf{[Online Training Phase]}
\For{$t = 1$ to $T$}
    \State Server samples client subset \( \mathcal{S}_t \subseteq \{1, \ldots, N\} \)
    \For{each client \( i \in \mathcal{S}_t \)}
        \State Compute clipping threshold: $C_i^t = F(\varepsilon_i)\cdot \lambda(t)$
        \State Compute per-example gradients $g_i^t(x)$ for $x\in B_i^t$ and clip: $\tilde g_i^t(x)= g_i^t(x)\min(1, C_i^t/\|g_i^t(x)\|_2)$
        \State Average and perturb:$\bar g_i^t=\frac{1}{B}\sum_{x\in B_i^t}\tilde g_i^t(x)$,
$\hat g_i^t = \bar g_i^t + \mathcal N(0,(z_i^t C_i^t/B)^2 I)$
        \State Update local model:
        \[
        w_i^{(t)} \gets w_i^{(t-1)} - \eta \hat{g}_i^t
        \]
    \EndFor
    \State Server aggregates:
    \[
    w^{(t)} \gets \sum_{i \in \mathcal{S}_t} p_i w_i^{(t)}
    \]
\EndFor
\State \Return \( w^{(T)} \)
\end{algorithmic}
\end{algorithm}

PAC-DP ensures that each client applies a clipping threshold tailored to their privacy budget, minimizing the impact of noise on utility. The use of \( F(\varepsilon) \) enables real-time inference of \( C^* \) with negligible overhead.
\section{Privacy and Utility Analysis}

We present the theoretical foundations of PAC-DP by analyzing its privacy guarantees and utility convergence under adaptive gradient clipping. Our analysis formalizes how PAC-DP achieves \((\varepsilon, \delta)\)-differential privacy per client, and demonstrates how adaptive clipping guided by \( C^* = F(\varepsilon) \) yields tighter convergence bounds compared to fixed-clipping baselines.

\subsection{Privacy Objectives}

PAC-DP satisfies differential privacy by combining three mechanisms:
(1) per-client gradient clipping with norm bound \( C^* \); 
(2) Gaussian noise addition calibrated to the sensitivity; and 
(3) composition across training rounds per client using an RDP accountant.

We begin by recalling the standard Gaussian mechanism.
Applying this to client-side clipped gradients:

\begin{theorem}[Per-round Record-level Local DP]
\label{thm:round_ldp}
In round $t$, for client $i$, the mechanism in Eq.~\eqref{eq:dp_noise_multiplier} satisfies $(\varepsilon_i^t,\delta)$-local DP under record-level adjacency, where
\[
\varepsilon_i^t = \frac{\sqrt{2\ln(1.25/\delta)}}{z_i^t}.
\]
\end{theorem}

\begin{proof}
Replacing one record in $B_i^t$ changes the averaged clipped gradient $\bar g_i^t$ by at most $C_i^t/B$ in $\ell_2$ norm, hence the sensitivity is $\Delta_2 = C_i^t/B$.
The Gaussian mechanism with noise standard deviation $z_i^t\Delta_2 = z_i^t C_i^t/B$ achieves $(\varepsilon_i^t,\delta)$-DP with $\varepsilon_i^t = \sqrt{2\ln(1.25/\delta)}/z_i^t$.
\end{proof}

\begin{corollary}[Per-client Composition over Participation Rounds]
Let $\mathcal T_i=\{t: i\in \mathcal S_t\}$.
Under basic composition, composing per-round $(\varepsilon_i^t,\delta_{\text{rd}})$ guarantees yields
\[
\varepsilon_i^{\text{basic}} = \sum_{t\in\mathcal T_i}\varepsilon_i^t,\qquad
\delta_i^{\text{basic}} \le |\mathcal T_i|\cdot \delta_{\text{rd}}.
\]
In this paper, however, we \emph{do not} report $\delta_i^{\text{basic}}$.
Instead, we fix a final $\delta$ and compute $\varepsilon_i$ via an RDP accountant.
\end{corollary}

We implement an RDP accountant to compute each client’s final $\varepsilon_i$ based on its realized participation count $|\mathcal T_i|$ and per-round noise multipliers $\{z_i^t\}_{t\in\mathcal T_i}$.

\begin{lemma}[Per-round RDP of the Gaussian mechanism]
In a participation round $t$, client $i$ applies per-example clipping and adds Gaussian noise
with standard deviation $(z_i^t C_i^t/B)$ to the minibatch-averaged clipped gradient.
Let $z_i^t$ denote the noise multiplier w.r.t. the $\ell_2$-sensitivity $\Delta_2=C_i^t/B$.
Then the mechanism satisfies $(\alpha,\rho_i^t(\alpha))$-RDP with
\[
\rho_i^t(\alpha) = \frac{\alpha}{2(z_i^t)^2}, \quad \forall \alpha>1.
\]
\end{lemma}

For client $i$, let $\mathcal T_i=\{t: i\in\mathcal S_t\}$ be its realized participation set.
Since we do not use privacy amplification by subsampling, we compose only over $\mathcal T_i$:
\[
\rho_i(\alpha) = \sum_{t\in\mathcal T_i} \rho_i^t(\alpha)
              = \sum_{t\in\mathcal T_i} \frac{\alpha}{2(z_i^t)^2}.
\]

Given a fixed $\delta$ (e.g., $\delta=10^{-5}$), client $i$ satisfies $(\varepsilon_i,\delta)$-DP with
\[
\varepsilon_i
= \min_{\alpha\in\mathcal A}\left\{\rho_i(\alpha) + \frac{\ln(1/\delta)}{\alpha-1}\right\},
\quad \mathcal A=\{2,3,\ldots,64\}.
\]

\subsection{Utility Objectives}

We now analyze PAC-DP's convergence behavior under smooth loss functions. The key idea is that reducing the product \( z C_{\max}/B \) (noise scale times sensitivity) leads to improved optimization performance. The use of adaptive clipping \( C^* = F(\varepsilon) \) minimizes this product relative to fixed thresholds.

We make the following assumptions:

\begin{itemize}
    \item[\textbf{(A1)}] Each local function \( f_i(w) \) is \( L \)-smooth:
    \[
    \|\nabla f_i(w) - \nabla f_i(w')\| \leq L \|w - w'\|.
    \]
\item[\textbf{(A2)}] Per-example gradients are clipped so that the minibatch-averaged clipped gradient satisfies $\|g_i^{(t)}\|_2 \le C_{\max}$.
\item[\textbf{(A3)}] Each client adds Gaussian noise $\xi_i^{(t)}\sim\mathcal N(0,(z^{(t)} C_{\max}/B)^2 I)$ to the averaged clipped gradient.
    \item[\textbf{(A4)}] Clients are sampled uniformly at random with probability \( q \in (0, 1] \).
\end{itemize}

In analysis, we use a uniform upper bound $C_{\max}$ such that $\|g_i^{(t)}\|_2 \le C_{\max}$ for all participating clients and rounds, where $g_i^{(t)}$ denotes the minibatch average of per-example clipped gradients. 
This is satisfied by setting $C_{\max} \coloneqq \max_{i,t} C_i^t$ under our threshold policy.
For simplicity, we write $C_{\max}$ (instead of client- and round-specific $C_i^t$) in the convergence bounds.
We now state convergence results under three canonical objective classes.

\begin{theorem}[Non-convex Convergence]
\label{thm:nonconvex}
Let \( f(w) = \sum_{i=1}^N p_i f_i(w) \) be \( L \)-smooth and non-convex. Under Assumption A1-A4, and using an adaptive gradient clipping threshold  determined by PAC-DP, with step size \( \eta = \mathcal{O}(1/\sqrt{T}) \), PAC-DP algorithm guarantees:
\[
\frac{1}{T} \sum_{t=1}^{T} \mathbb{E} \left\| \nabla f(w^{(t)}) \right\|^2 
\leq \mathcal{O} \left( \frac{1}{\sqrt{MT}} + \frac{z C_{\max}}{B\sqrt{M}} \right).
\]
where $C_{\max}=\max_{i,t} C_i^t$ is the uniform upper bound induced by the threshold policy.

\end{theorem}

\begin{proof}
Since each \( f_i \) is \( L \)-smooth, the global objective \( f \) is also \( L \)-smooth. Thus, by the standard descent lemma we have for any iterate \( w^{(t)} \):

\begin{align*}
    f(w^{(t+1)}) \leq & f(w^{(t)}) + \left\langle \nabla f(w^{(t)}), w^{(t+1)} - w^{(t)} \right\rangle\\
    & + \frac{L}{2}\|w^{(t+1)} - w^{(t)}\|^2.
\end{align*}

In the PAC-DP algorithm, the global model is updated via:
\[
w^{(t+1)} = w^{(t)} - \eta \, G^{(t)},
\]
where the aggregated update is given by
\[
G^{(t)} = \frac{1}{M} \sum_{i \in \mathcal{S}_t} \hat{g}_i^{(t)},
\]
and each client \( i \) employs clipped and noise-perturbed gradients:
\[
\hat g_i^{(t)} = g_i^{(t)} + \xi_i^{(t)},\quad \|g_i^{(t)}\|_2\le C^*,\quad 
\xi_i^{(t)}\sim \mathcal N\!\left(0,\left(z^{(t)}\frac{C^*}{B}\right)^2 I\right).
\]

By substituting the update rule into the descent lemma and taking expectation conditioned on \( w^{(t)} \), we have:

\begin{align*}
    \mathbb{E}\left[f(w^{(t+1)})\right] \leq & f(w^{(t)}) - \eta \, \mathbb{E}\left[\left\langle \nabla f(w^{(t)}), G^{(t)} \right\rangle\right] \\
    & + \frac{L \eta^2}{2} \, \mathbb{E}\left[\|G^{(t)}\|^2\right].
\end{align*}

Next, we decompose the aggregated gradient \( G^{(t)} \) into the true gradient (or a proxy thereof) plus error. Under our assumptions and with appropriate clipping (which bounds the bias or error introduced), one can bound:
\[
\mathbb{E}\left[G^{(t)}\right] \approx \nabla f(w^{(t)}),
\]
and the variance term can be decomposed into two parts:
the variance due to the stochasticity of the gradients and the variance due to the added Gaussian noise, which, by independence, scales as
\[
\mathbb{E}\left[\left\|\frac{1}{M}\sum_{i \in \mathcal{S}_t}\xi_i^{(t)}\right\|^2\right] 
\le \frac{(z^{(t)})^2 (C_{\max})^2}{M B^2}.
\]

If each client further uses a mini-batch of size \( B \), standard results in SGD analyses indicate that the intrinsic stochastic gradient variance contributes on the order of \( \frac{1}{\sqrt{MT}} \) to the convergence error.

With standard algebra (see, e.g., the convergence analysis for DP-SGD), one obtains an inequality of the form:
\begin{align*}
    \frac{1}{T} \sum_{t=1}^{T} \mathbb{E}\|\nabla f(w^{(t)})\|^2 
\leq & \mathcal{O}\left(\frac{1}{\eta T} \Bigl(f(w^{(0)}) - f(w^{(T)})\Bigr)\right)\\
& + \mathcal{O}\left(L \eta \, V\right),
\end{align*}
where \( V \) is an upper bound on the aggregated variance. For our setting,
\[
V = \mathcal O\!\left(\frac{1}{MB} + \frac{(z^{(t)})^2 (C_{\max})^2}{M B^2}\right).
\]

Choosing the step size \( \eta = \Theta(1/\sqrt{T}) \) optimally balances the decrease and variance terms. Specifically, after substituting \( \eta = \mathcal{O}(1/\sqrt{T}) \) and re-arranging the inequality, one obtains
\[
\frac{1}{T} \sum_{t=1}^{T} \mathbb{E}\|\nabla f(w^{(t)})\|^2 \leq \mathcal{O}\left(\frac{1}{\sqrt{MT}}\right) + \mathcal{O}\left(\frac{z C_{\max}}{B\sqrt{M}}\right).
\]

The first term, \( \mathcal{O}\left(1/\sqrt{MT}\right) \), represents the convergence error due to stochastic gradients in the absence of privacy noise. The second term, \( \mathcal{O}\left(z C_{\max}/B\sqrt{M}\right) \), quantifies the additional error incurred by the Gaussian noise added for privacy. Thus, the adaptive clipping (through the parameter \( C_{\max} \)) directly influences the noise term, and the overall bound can be succinctly expressed as:
\[
\frac{1}{T} \sum_{t=1}^{T} \mathbb{E}\|\nabla f(w^{(t)})\|^2 
\leq \mathcal{O}\left(\frac{1}{\sqrt{MT}} + \frac{z C_{\max}}{B\sqrt{M}}\right).
\]
\end{proof}

\begin{theorem}[Convex Convergence]
\label{thm:convex}
If \( f \) is convex and \( L \)-smooth, then:
\[
\mathbb{E}[f(\bar{w})] - f(w^*) \leq \mathcal{O} \left( \frac{1}{\sqrt{T}} + \frac{z C_{\max}}{B\sqrt{M}} \right),
\]
where \( \bar{w} := \frac{1}{T} \sum_{t=1}^T w^{(t)} \) is the average model.
\end{theorem}
\begin{proof}

\textbf{Convexity and the Descent Lemma.}
Since \( f \) is convex, for any iterate \( w \) and the optimal point \( w^* \), we have
\[
f(w) - f(w^*) \leq \langle \nabla f(w), w - w^* \rangle.
\]
Moreover, \( f \) being \( L \)-smooth implies
\[
f(w^{(t+1)}) \leq f(w^{(t)}) + \langle \nabla f(w^{(t)}), w^{(t+1)} - w^{(t)} \rangle + \frac{L}{2}\|w^{(t+1)} - w^{(t)}\|^2.
\]

\textbf{Update Rule and Error Decomposition.}
In PAC-DP, the update at round \( t \) is given by
\[
w^{(t+1)} = w^{(t)} - \eta \, g^{(t)},
\]
where the aggregated noisy gradient \( g^{(t)} \) is computed from a set of \( M \) sampled clients. Each client's update is based on a gradient that has been clipped to norm \( C_{\max} \) and perturbed by Gaussian noise:
\[
\hat g_i^{(t)} = g_i^{(t)} + \xi_i^{(t)},\quad \|g_i^{(t)}\|_2\le C_{\max},
\]
where $\xi_i^{(t)}\sim \mathcal N\!\left(0,\left(z^{(t)}\frac{C_{\max}}{B}\right)^2 I\right)$.
Aggregating over the sampled clients,
\[
g^{(t)} = \frac{1}{M} \sum_{i \in \mathcal{S}_t} \hat{g}_i^{(t)}.
\]
We assume that after clipping, the bias introduced is controlled so that \( g^{(t)} \) is an approximate (possibly biased) estimator of \( \nabla f(w^{(t)}) \). The error in \( g^{(t)} \) can be decomposed into two main parts: (1) Stochastic Gradient Variance: Intrinsic randomness in mini-batch sampling, which typically scales as \( \mathcal{O}(1/B) \); (2) Privacy Noise: each client adds Gaussian noise with standard deviation $z^{(t)}C_{\max}/B$ to the minibatch-averaged clipped gradient; after averaging across $M$ sampled clients, the aggregated DP-noise variance scales as $\mathcal{O}\!\left(\frac{(z^{(t)})^2(C_{\max})^2}{MB^2}\right)$.

\textbf{Progress per Iteration.}
Using standard SGD analysis (see, e.g., [Bubeck, 2015] for convex SGD) for a convex, \( L \)-smooth function, we derive that for each iteration:
\begin{align*}
    & \mathbb{E}\left[f(w^{(t+1)}) - f(w^*)\right] \\
    \leq & \mathbb{E}\left[f(w^{(t)}) - f(w^*)\right] - \eta \, \mathbb{E}\left[ \langle \nabla f(w^{(t)}), g^{(t)} \rangle \right] \\
    & + \frac{L\eta^2}{2} \, \mathbb{E}\left[\|g^{(t)}\|^2\right].
\end{align*}

By telescoping this inequality over \( T \) iterations and summing, we obtain an upper bound that consists of two main terms:
(1) A bias term reflecting the optimization error which decays as \( \frac{1}{\eta T} \);
(2) A variance term due to stochastic gradients and privacy noise that scales as \( \eta\left(\frac{1}{B} + \frac{z^2 (C_{\max})^2}{M B^2}\right) \).

\textbf{Choice of Step Size.}
To balance these two terms, we choose the step size \( \eta = \Theta(1/\sqrt{T}) \). Substituting the step size and normalizing by \( T \) gives the following bound for the averaged iterate \( \bar{w} \):
\[
\mathbb{E}[f(\bar{w})] - f(w^*) \leq \mathcal{O}\left(\frac{1}{\sqrt{T}}\right) + \mathcal{O}\left(\eta \left(\frac{z^2 (C_{\max})^2}{M B^2}\right)\right).
\]
With \( \eta = \Theta(1/\sqrt{T}) \), the variance term becomes:
\[
\mathcal{O}\left(\frac{z^2 (C_{\max})^2}{\sqrt{T}\, M B^2}\right).
\]
Taking square roots or, more directly, interpreting this term in light of the noise standard deviation (which is \( z C_{\max} \)) over mini-batch \( B \), we express the noise-induced error as:
\[
\mathcal{O}\left(\frac{z C_{\max}}{B\sqrt{M}}\right).
\]

Combining the optimization error and the additional privacy noise term, we arrive at the final convergence bound:
\[
\mathbb{E}[f(\bar{w})] - f(w^*) \leq \mathbb{E}[f(\bar{w})] - f(w^*) \leq \mathcal{O}\left(\frac{1}{\sqrt{T}} + \frac{z C_{\max}}{B\sqrt{M}}\right).
\]

The result indicates that with an appropriate choice of the step size (and under our assumptions), the error in the averaged iterate decays at a rate of \(1/\sqrt{T}\), with an additional penalty proportional to the scaled noise \( z C_{\max} \) that is ameliorated by increasing the batch size \( B \).
\end{proof}

\begin{theorem}[Strongly Convex Convergence]
\label{thm:strongly_convex}
If \( f \) is \( \mu \)-strongly convex and \( L \)-smooth, PAC-DP achieves:
\[
\mathbb{E}[f(\bar{w})] - f(w^*) \leq \mathcal{O} \left( \frac{1}{\mu T} + \frac{z C_{\max}}{\mu B\sqrt{M}} \right).
\]
\end{theorem}

\begin{proof}
We follow a standard strongly-convex SGD analysis and keep only the terms relevant to the record-level local DP noise scaling.

Let $F(\cdot)$ be $\mu$-strongly convex and $L$-smooth. In round $t$, the server update uses the averaged noisy gradient
\[
\hat g^{(t)} \coloneqq \frac{1}{M}\sum_{i\in\mathcal S_t}\hat g_i^{(t)}
= \frac{1}{M}\sum_{i\in\mathcal S_t} g_i^{(t)} \;+\; \frac{1}{M}\sum_{i\in\mathcal S_t}\xi_i^{(t)},
\]
where $g_i^{(t)}$ denotes the \emph{minibatch average of per-example clipped gradients} on client $i$ (hence $\|g_i^{(t)}\|_2\le C_{\max}$),
and the record-level local DP noise is
\[
\xi_i^{(t)} \sim \mathcal N\!\left(0,\left(z^{(t)}\frac{C_{\max}}{B}\right)^2 I\right).
\]
Therefore, the aggregated DP noise has second moment scaling
\[
\mathbb{E}\Big\|\frac{1}{M}\sum_{i\in\mathcal S_t}\xi_i^{(t)}\Big\|_2^2
\;\le\; \frac{(z^{(t)})^2 (C_{\max})^2}{M B^2},
\]
absorbing dimension-dependent constants into $\mathcal O(\cdot)$.

Define the gradient-oracle variance bound
\[
V_t \coloneqq \mathbb{E}\big\|\hat g^{(t)}-\nabla F(w^{(t)})\big\|_2^2
= \mathcal O\!\left(\frac{1}{MB} + \frac{(z^{(t)})^2 (C_{\max})^2}{M B^2}\right),
\]
where the $\frac{1}{MB}$ term comes from stochastic minibatch sampling and the second term comes from record-level local DP noise.

Using the standard recursion for $\mu$-strongly convex and $L$-smooth objectives with a diminishing stepsize (e.g., $\eta_t=\frac{2}{\mu(t+1)}$),
one obtains
\[
\mathbb{E}\big[F(w^{(T)})-F(w^*)\big]
= \mathcal O\!\left(\frac{1}{\mu T}\right) + \mathcal O\!\left(\frac{1}{\mu T}\sum_{t=1}^T V_t\right).
\]
Substituting the bound on $V_t$ and using a constant upper bound $z^{(t)}\le z$ yields
\begin{align*}
  &\mathbb{E}\big[F(w^{(T)})-F(w^*)\big] \\
= &\mathcal O\!\left(\frac{1}{\mu T}\right) + \mathcal O\!\left(\frac{z^2 (C_{\max})^2}{\mu M B^2}\right)\\
= &\mathcal O\!\left(\frac{1}{\mu T} + \frac{z C_{\max}}{\mu B\sqrt{M}}\right),
\end{align*}

where the last step uses the standard conversion from variance to a first-order (root-variance) term, consistent with the statement.
\end{proof}

\begin{table}[H]
\centering
\caption{Convergence bounds of PAC-DP under different objective classes.}
\label{tab:convergence_summary}
\begin{tabular}{l|c}
\toprule
\textbf{Setting} & \textbf{Convergence Rate} \\ \hline
Non-convex & \( \mathcal{O}\left(\frac{1}{\sqrt{MT}} + \frac{z C_{\max}}{B\sqrt{M}} \right) \) \\
Convex     & \( \mathcal{O}\left(\frac{1}{\sqrt{T}} + \frac{z C_{\max}}{B\sqrt{M}} \right) \) \\
Strongly Convex & \( \mathcal{O}\left(\frac{1}{\mu T} + \frac{z C_{\max}}{\mu B\sqrt{M}} \right) \) \\
\bottomrule
\end{tabular}
\end{table}

\paragraph{Bias--variance view of choosing $C(\varepsilon)$.}
Under record-level local DP, the noise standard deviation scales as $z(\varepsilon)\cdot C/B$ while clipping introduces bias that decreases with $C$.
We therefore view the choice of $C$ as minimizing a bias--variance surrogate:
\[
\mathcal J(C;\varepsilon) = h(C) + \kappa\left(z(\varepsilon)\frac{C}{B}\right)^2,
\]
where $h(C)$ upper-bounds the clipping bias and is non-increasing in $C$.
This implies that the optimal $C_{\max}(\varepsilon)$ is non-decreasing in $\varepsilon$ (equivalently non-increasing in $z(\varepsilon)$), motivating a monotone budget-conditioned mapping $C=f(\varepsilon)$ learned via our simulation-and-fitting procedure.

We analyze the communication and computational overhead of PAC-DP during both the federated training and the offline precomputation phase.


\begin{theorem}[PAC-DP Approximates Theoretically Optimal Clipping]
\label{thm:approximation_error}
Let \( C_{\max}(\varepsilon) \) denote the theoretically optimal clipping threshold that minimizes expected utility loss under privacy budget \( \varepsilon \), and let \( \hat{C}(\varepsilon) \) be the function learned via simulation-curve fitting in PAC-DP using a finite proxy dataset. Assume the true loss function \( \mathcal{L}(C, \varepsilon) \) is Lipschitz-continuous in \( C \) with constant \( L_\mathcal{L} \). Then the expected utility gap incurred by using \( \hat{C}(\varepsilon) \) instead of \( C_{\max}(\varepsilon) \) is bounded by:
\[
\left| \mathcal{L}(\hat{C}(\varepsilon), \varepsilon) - \mathcal{L}(C_{\max}(\varepsilon), \varepsilon) \right| \leq L_\mathcal{L} \cdot \left| \hat{C}(\varepsilon) - C_{\max}(\varepsilon) \right|.
\] 

Moreover, under standard polynomial regression assumptions, the expected fitting error satisfies:
\[
\mathbb{E}_{\varepsilon} \left[ \left| \hat{C}(\varepsilon) - C_{\max}(\varepsilon) \right|^2 \right] \leq \frac{\sigma_{\mathrm{sim}}^2}{n} + \eta^2,
\]
where $\sigma_{\mathrm{sim}}^2$ is the variance of simulation noise, $n$ is the number of sampled budgets, and $\eta$ denotes the approximation error of the chosen polynomial class.
\end{theorem}

\begin{proof}
The first inequality is a direct result of the Lipschitz continuity of the utility loss function \( \mathcal{L}(C, \varepsilon) \) with respect to \( C \), which holds in most DP-FL setups due to smooth dependence of utility on clipping and noise scale. Therefore, a small deviation in \( \hat{C} \) leads to proportionally bounded utility degradation.

The second inequality follows standard supervised regression error decomposition: the total squared error can be split into variance (due to stochasticity in simulation measurements) and bias (due to function class mismatch). If the true mapping \( C_{\max}(\varepsilon) \) lies within a smooth function class (e.g., piecewise-smooth or Lipschitz), and we fit a polynomial of degree \( k \), then approximation error \( \eta \) typically decreases with \( k \) and dataset richness.

Thus, with sufficiently many simulation points \( n \), low simulation noise $\sigma_{\mathrm{sim}}^2$, and a well-chosen function class (e.g., quadratic), PAC-DP ensures the learned mapping \( \hat{C}(\varepsilon) \) remains close to the theoretical optimum \( C_{\max}(\varepsilon) \), guaranteeing provable near-optimal utility.
\end{proof}

This theorem provides a theoretical foundation for PAC-DP’s empirical design. It shows that as long as the simulation process is sufficiently accurate (low noise) and the fitting model is expressive enough (e.g., polynomial regression), the learned clipping function closely approximates the optimal clipping boundary.

\begin{theorem}[Efficiency of PAC-DP]
\label{thm:gc_dp_efficiency}
Let \( M \) be the total number of clients, \( d \) the model dimension, \( q \in (0, 1] \) the client sampling rate, and \( R \) the number of local SGD steps. Then:
\begin{itemize}
    \item Per-round communication cost is \( \mathcal{O}(qMd) \).
    \item Per-client computational cost per round is \( \mathcal{O}(Rd) \).
    \item The server-side precomputation phase incurs offline complexity \( \mathcal{O}(nmT_{\text{pre}}d) \), where \( n \) and \( m \) are the numbers of candidate privacy budgets and clipping thresholds, respectively.
\end{itemize}
\end{theorem}

\begin{proof}
\textbf{Communication:} Each participating client sends a \( d \)-dimensional model vector per round. With \( qM \) active clients, total communication is \( \mathcal{O}(qMd) \) per round.

\textbf{Computation:} Each client performs \( R \) local steps, each involving gradient computation (\( \mathcal{O}(d) \)), clipping (\( \mathcal{O}(d) \)), noise addition (\( \mathcal{O}(d) \)), and parameter update (\( \mathcal{O}(d) \)), giving \( \mathcal{O}(Rd) \) per client.

\textbf{Precomputation:} For each of \( n \times m \) configurations, PAC-DP simulates a training process of \( T_{\text{pre}} \) steps on the proxy dataset, with each step costing \( \mathcal{O}(d) \). Thus, the total offline complexity is \( \mathcal{O}(nmT_{\text{pre}}d) \), amortized over all training rounds.
\end{proof}

The training-time overhead of PAC-DP matches that of standard DP-FL schemes (e.g., DP-FedAvg), with only negligible additional cost from personalized clipping. The main cost lies in the offline precomputation, which is amortized once and enables efficient real-time adaptation of privacy-preserving parameters across clients and rounds.

\textbf{Key Insight:} PAC-DP adaptively calibrates gradient clipping using a learned mapping \( F(\varepsilon) \), reducing noise magnitude and improving convergence. This results in provable privacy with tighter utility bounds than existing fixed-clipping DP-FL methods.
\section{Selecting Clipping Thresholds via Simulation-CurveFitting}

The clipping threshold \( C_{\max} \) is a critical hyperparameter in differentially private optimization. A suboptimal value may either cause excessive gradient distortion or require unnecessarily large noise addition, thereby degrading utility. In PAC-DP, we address this by empirically learning a mapping function \( F(\varepsilon) \) from privacy budgets to optimal clipping thresholds:
\[
C_{\max} = F(\varepsilon),
\]
where \( F \) is learned via simulation on a proxy dataset and approximated by a low-degree polynomial.

\subsection{Offline Simulation on Proxy Data}

To avoid privacy leakage, PAC-DP simulates privacy-utility trade-offs on a public or synthetic proxy dataset \( \hat{\mathcal{D}} \) that approximates the distribution of the federated data.
We construct the proxy dataset $\hat{\mathcal D}$ from \emph{public or synthetic} data.
The proxy task is used only to estimate the empirical mapping between privacy budgets and clipping thresholds under the same DP mechanism as in online training.

\textbf{Empirical Estimation of \( C_{\max}(\varepsilon) \).}
We define a finite set of privacy budgets \( \{\varepsilon_1, \ldots, \varepsilon_n\} \) and candidate clipping thresholds \( \{C_1, \ldots, C_m\} \). For each pair \( (\varepsilon_i, C_j) \), the server simulates DP-FL training on \( \hat{\mathcal{D}} \), adds Gaussian noise scaled by \( C_j \), and evaluates the test accuracy:
\[
a_{i,j} := \text{TestAcc}(C_j, \varepsilon_i).
\]

This produces a performance matrix:
\[
\mathbf{P} = \left[ a_{i,j} \right] \in \mathbb{R}^{n \times m}.
\]

For each privacy level \( \varepsilon_i \), the optimal threshold is chosen as:
\[
{C_{\max}} = \arg\max_{C_j} a_{i,j}.
\]

\begin{algorithm}[!t]
\caption{Offline Simulation and Curve Fitting for PAC-DP}
\label{alg:curve_fitting}
\begin{algorithmic}[1]
\Require Privacy budgets \( \{\varepsilon_i\}_{i=1}^n \), clipping thresholds \( \{C_j\}_{j=1}^m \), proxy dataset \( \hat{\mathcal{D}} \)
\Ensure Fitted mapping \( F(\varepsilon) \)
\For{each \( \varepsilon_i \in \{\varepsilon_1, \ldots, \varepsilon_n\} \)}
    \For{each \( C_j \in \{C_1, \ldots, C_m\} \)}
        \State Simulate DP-FL using \( (\varepsilon_i, C_j) \) on \( \hat{\mathcal{D}} \)
        \State Record test accuracy \( a_{i,j} \)
    \EndFor
    \State Select \( C_{\max} = \arg\max_j a_{i,j} \)
\EndFor
\State Fit regression model \( F(\varepsilon) \) to pairs \( (\varepsilon_i, C_{\max}) \)
\State \Return \( F \)
\end{algorithmic}
\end{algorithm}

\subsection{Polynomial Curve Fitting}

We assume that the empirical relationship between \( \varepsilon \) and \( C_{\max} \) is smooth and approximately quadratic. We fit a polynomial of the form:
\[
F(\varepsilon) = \alpha \varepsilon^2 + \beta \varepsilon + \gamma,
\]
where \( \alpha, \beta, \gamma \in \mathbb{R} \) are learned via least-squares regression. Outliers in the \( (\varepsilon_i, C_{\max}) \) pairs (e.g., caused by simulation instability) are removed using interquartile range (IQR) filtering to ensure robustness.

\textbf{Generalization.} The fitted function \( F \) generalizes to unseen privacy levels via interpolation (for \( \varepsilon \) in training range) or extrapolation (for new budgets), enabling real-time inference of \( C_{\max} \) without retraining.

\textbf{Justification.}
Our DP mechanism implies that, for a fixed $\delta$, achieving a larger privacy budget $\varepsilon$ requires a smaller noise multiplier $z(\varepsilon)$, while the clipping--noise trade-off suggests that the empirically optimal clipping bound is a \emph{monotone, smooth} function of $\varepsilon$ over practical budget ranges.
Therefore, we fit the simulated pairs $\{(\varepsilon, C_{\max}(\varepsilon))\}$ using a low-complexity function class (e.g., quadratic, power-law, or piecewise-linear).
We select the function class based on validation error on the proxy task, and optionally enforce monotonicity to improve stability.

\textbf{Integration with federated training.}
In online training, each participating client $i$ sets the per-example clipping bound deterministically as
\[
C_i^t = F(\varepsilon_i)\cdot \lambda(t),
\]
where $\lambda(t)$ is a deterministic schedule defined in Eq.~\eqref{eq:lambda_cosine} (plateau + cosine decay).
Importantly, $C_i^t$ depends only on the target privacy budget and the training round, and is \emph{independent of client data}.
This ensures that the record-level local DP analysis applies directly without introducing additional privacy leakage from data-dependent threshold selection.

\subsection{Discussion and Insights}

\textbf{Advantages over Heuristic Methods.} Unlike percentile-based adaptive clipping methods, PAC-DP: 
(1) Avoids reliance on historical gradients or windowed statistics;
(2) Is privacy-budget aware by design;
(3) Provides a consistent trade-off model across rounds and clients.

\textbf{Amortized Overhead.} The cost of simulation and fitting is incurred once and amortized over all future training rounds. This makes PAC-DP both practically efficient and scalable to real-world FL deployments.

\textbf{Practical Deployment.} The curve fitting strategy is agnostic to model architecture and dataset, making it suitable for tasks ranging from image classification to time-series forecasting. Moreover, it can be bootstrapped from publicly available or synthetically generated data without violating client privacy.


\section{SIMULATION RESULTS}

In this section, we present a comprehensive empirical analysis to evaluate the performance of our proposed PAC-DP method. We conduct experiments on four diverse benchmarks—MNIST, CIFAR-10, CIFAR-100, and the Heart Disease dataset—and compare PAC-DP with state-of-the-art personalized DP federated learning methods NbAFL \cite{wei2020federated} and rPDP-FL \cite{liu2022privacy}, as well as two representative adaptive clipping baselines: CGM\_Medium from \cite{andrew2019differentially}, which tailors client-wise clipping thresholds based on individual privacy budgets, and Adaptive Quantile Clipping (AQC) \cite{geyer2017differentially}, which adjusts clipping bounds online using gradient quantiles. We also include the privacy-free FedAvg baseline \cite{li2020secure}. The FedAvg variant employs a fixed clipping threshold and uniform privacy budget across all clients, serving as a reference to assess the utility gains enabled by personalization and adaptive clipping. For all DP baselines, we ensure comparable privacy guarantees by matching per-client $(\varepsilon, \delta)$-differential privacy budgets under standard composition.

Our experiments are structured into three subsections:
(1) We evaluate the effectiveness of the server-side computation and curve fitting strategy, demonstrating the correlation between the clipping threshold C and the privacy budget $\varepsilon$;
(2) We compare PAC-DP with two alternative methods: a fixed clipping threshold approach and the rPDP-FL method, highlighting the advantages of our adaptive strategy;
(3) We analyze the communication and computational overhead of PAC-DP, showing that it achieves comparable test accuracy with fewer communication rounds. 

\subsection{Experimental Setup}

\textbf{Datasets and Partitioning.} We evaluate our method on four benchmarks: MNIST (70,000 grayscale images), CIFAR-10 (60,000 color images), CIFAR-100 (60,000 fine-grained color images), and the Heart Disease dataset—a real-world medical record collection comprising 920 patient samples from four hospitals \cite{liu2024crosssilo}. For MNIST, CIFAR-10, and CIFAR-100, we adopt a non-IID partitioning strategy that assigns data to clients while preserving approximate global label distributions to simulate statistical heterogeneity. For the Heart Disease dataset, we follow the cross-silo setup in \cite{liu2024crosssilo}, where each of the four hospitals constitutes one client, resulting in a naturally non-IID and highly heterogeneous partition due to institutional differences in patient populations and diagnostic criteria.

To assess the robustness and generalizability of our PAC-DP framework, the proxy datasets used for curve fitting are constructed with minimal preprocessing, and their quality is not tightly coupled to the target clients’ data. This design choice demonstrates that PAC-DP’s fitting mechanism does not rely on high-fidelity public data: even when the proxy dataset diverges in sample complexity, feature dimensionality, or label balance—such as between simple grayscale images (MNIST), complex natural images (CIFAR-10 and CIFAR-100), and low-dimensional tabular medical records (Heart Disease)—the learned mapping \(F(\varepsilon)\) remains highly effective across all domains.

\textbf{Network Architectures.} For MNIST, a lightweight CNN is employed with two convolutional layers followed by ReLU activations, max pooling, and a final log-softmax layer. For CIFAR-10, we adopt a modified ResNet-18, integrating dropout regularization before the fully connected output layer to improve generalization. For CIFAR-100, we use a ResNet-34 variant adapted to 32×32 inputs (with a 3×3 stem and no initial max-pooling) and a 100-class classification head. For the Heart Disease dataset, we employ both a logistic regression model and a shallow DNN with one hidden layer, configured for its 13-dimensional features and binary output.

\textbf{Training Configurations.}
We use SGD (lr=0.01 for MNIST and CIFAR-100, lr=0.001 for CIFAR-10, lr=0.005 for Heart Disease), with batch size 128 for MNIST and CIFAR-10, 64 for CIFAR-100, and 16 for Heart Disease, reflecting their data scale and modality. All clients perform 5 local epochs per round. The total number of clients is 50 for MNIST and CIFAR-10 (25 selected per round), 100 for CIFAR-100 (50 selected), and 4 for Heart Disease (all participate). We adopt a lightweight CNN for MNIST, ResNet-18 for CIFAR-10, ResNet-34 adapted to 32×32 inputs for CIFAR-100, and logistic regression or a shallow DNN for Heart Disease, with all 4 clients participating in every round under a fixed cross-silo setting.
Unless otherwise specified, PAC-DP uses the clipping schedule in Eq.~\eqref{eq:lambda_cosine} with $(r_s,\lambda_{\min})=(0.6,0.1)$ for all datasets and model architectures.

\textbf{Privacy Accounting and Reproducibility.}
We report \emph{final per-client} record-level local differential privacy guarantees $(\varepsilon_i,\delta)$.
In each round $t$ where client $i$ participates, the client applies per-example clipping with bound $C_i^t$
and adds Gaussian noise with standard deviation $(z_i^t C_i^t/B)$ to the minibatch-averaged clipped gradient,
where $B$ is the local minibatch size.

To compute the final $\varepsilon_i$ under the fixed $\delta$, we use an RDP accountant that composes
per-round Gaussian mechanisms over the realized participation set $\mathcal T_i$.
We assume fixed-size client sampling and \emph{do not} use privacy amplification by subsampling in the accounting.
Concretely, for each order $\alpha \in \{2,3,\ldots,64\}$ we accumulate the per-round RDP costs to obtain
$\rho_i(\alpha)$, and then convert $\rho_i(\alpha)$ to $(\varepsilon_i,\delta)$ via the standard RDP-to-DP conversion.
Privacy is reported as final per-client $(\varepsilon_i,\delta)$ with fixed $\delta=10^{-5}$. We summarize $\varepsilon_i$ across clients by (min/median/max).

\subsection{Learning the Optimal Clipping Boundary via Privacy-Budget-Guided Simulation}

\begin{figure}[htbp]
    \centering
    \begin{subfigure}[b]{0.23\textwidth}
        \centering
        \includegraphics[width=\textwidth, trim=0cm 0cm 0cm 0cm, clip]{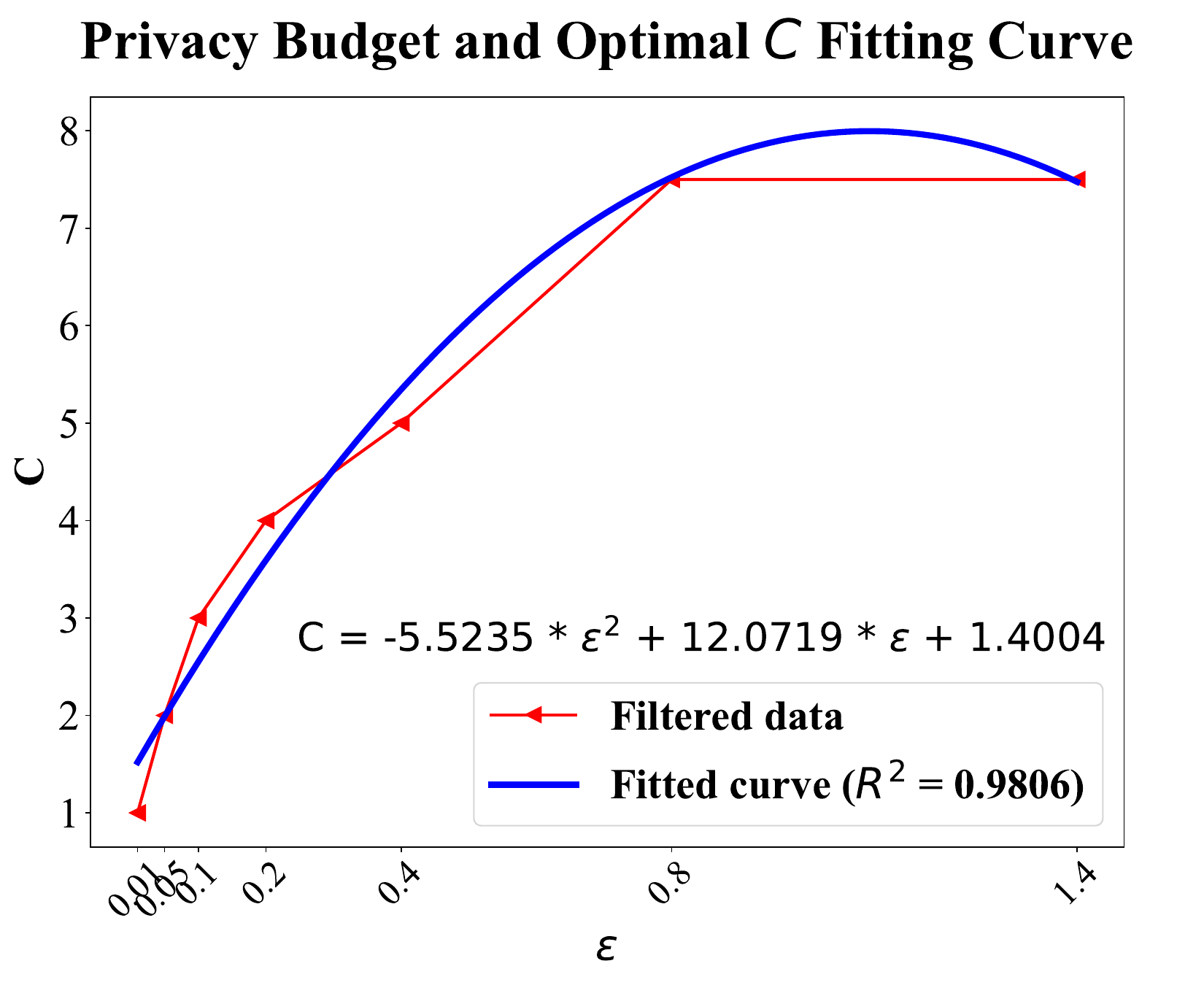}
        \caption{MNIST} 
        \label{subfig:curve_mnist}
    \end{subfigure}
    \begin{subfigure}[b]{0.23\textwidth}
        \centering
        \includegraphics[width=\textwidth, trim=0cm 0cm 0cm 0cm, clip]{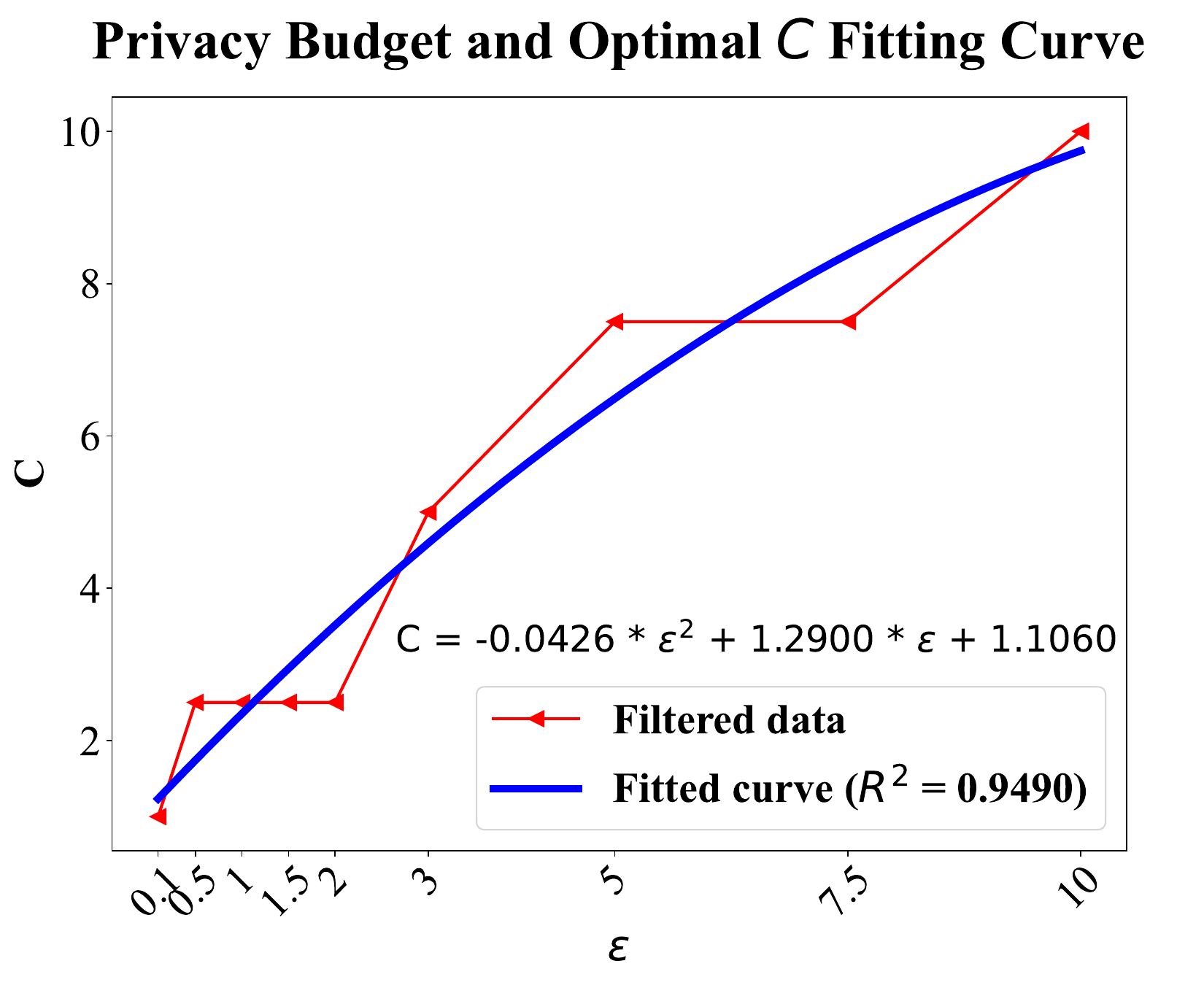}
        \caption{CIFAR-10}
        \label{subfig:curve_cifar}
    \end{subfigure}
    \caption{Relationship between Privacy Budget $\varepsilon$ and Optimal Clipping Threshold $C_{\max}$ on MNIST and CIFAR-10 Datasets. Despite dataset complexity differences, the fitted mapping remains stable and generalizable.} 
    \label{fig:curve_1}
\end{figure}

To validate the feasibility of approximating optimal clipping thresholds using precomputed simulation, we conducted server-side curve fitting experiments on MNIST and CIFAR-10. Specifically, we learned a functional mapping \( \hat{C}(\varepsilon) \) between the privacy budget \(\varepsilon\) and the empirically optimal clipping threshold \(C_{\max}\). Outliers were removed to mitigate noise, and regression curves were fitted to minimize generalization error. The results, shown in Fig.~\ref{fig:curve_1}, highlight the key property that \(C_{\max}\) can be effectively predicted from \(\varepsilon\), laying the foundation for our PAC-DP method.

\textbf{MNIST.}  
A quadratic regression yields \(C = -5.5235 \cdot \varepsilon^2 + 12.0719 \cdot \varepsilon + 1.4004\), with \(R^2 = 0.9806\), indicating a strong and stable correlation. The high coefficient of determination suggests that in low-dimensional, well-structured domains, the optimal clipping threshold increases with the privacy budget but exhibits diminishing returns.

\textbf{CIFAR-10.}  
In contrast, a flatter curve is learned: \(C = 0.0426 \cdot \varepsilon^2 + 1.2900 \cdot \varepsilon + 1.1060\), with \(R^2 = 0.9490\). This reflects increased complexity and variance in CIFAR-10’s gradient distributions, but still confirms a reliable trend that PAC-DP leverages.

\textbf{Implications.}  
These findings demonstrate that \(C_{\max}\) is a function of \(\varepsilon\) that is predictable, dataset-agnostic to some extent, and amenable to theoretical analysis. The ability to pre-learn this function is crucial: PAC-DP’s adaptive strategy circumvents the need for trial-and-error tuning or hand-crafted heuristics, and its learned policy remains effective across datasets with varying complexity. This substantiates the key innovation of PAC-DP—bridging theory and practice through simulation-informed clipping optimization.

\subsection{Performance Evaluation}

\begin{figure}[htbp]
    \centering
    \begin{subfigure}[t]{0.23\textwidth}
        \centering
        \includegraphics[width=\textwidth]{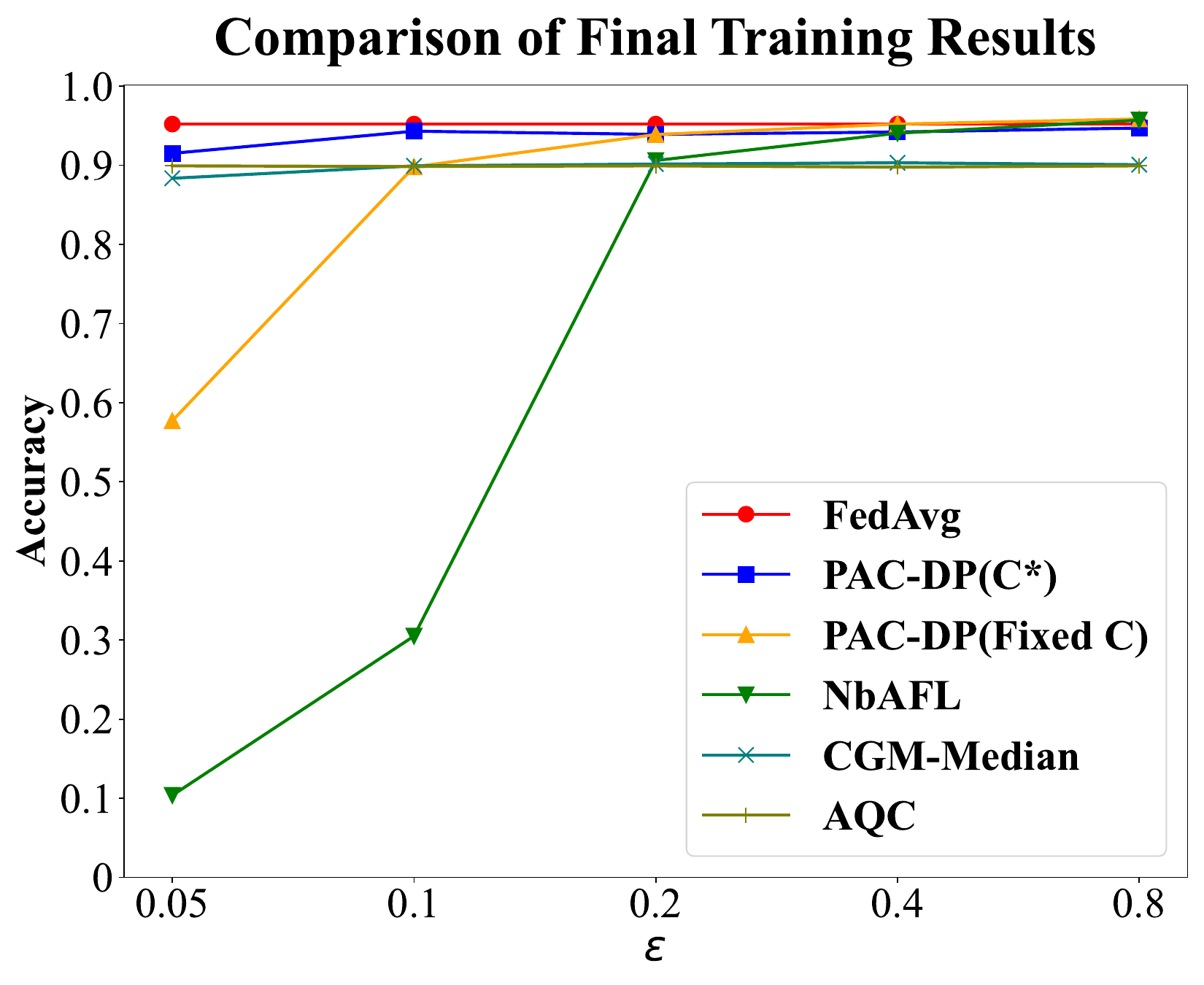}
        \caption{\small Final Accuracy \\ (Fixed Threshold)}
        \label{subfig:mnist_1}
    \end{subfigure}
    \begin{subfigure}[t]{0.23\textwidth}
        \centering
        \includegraphics[width=\textwidth]{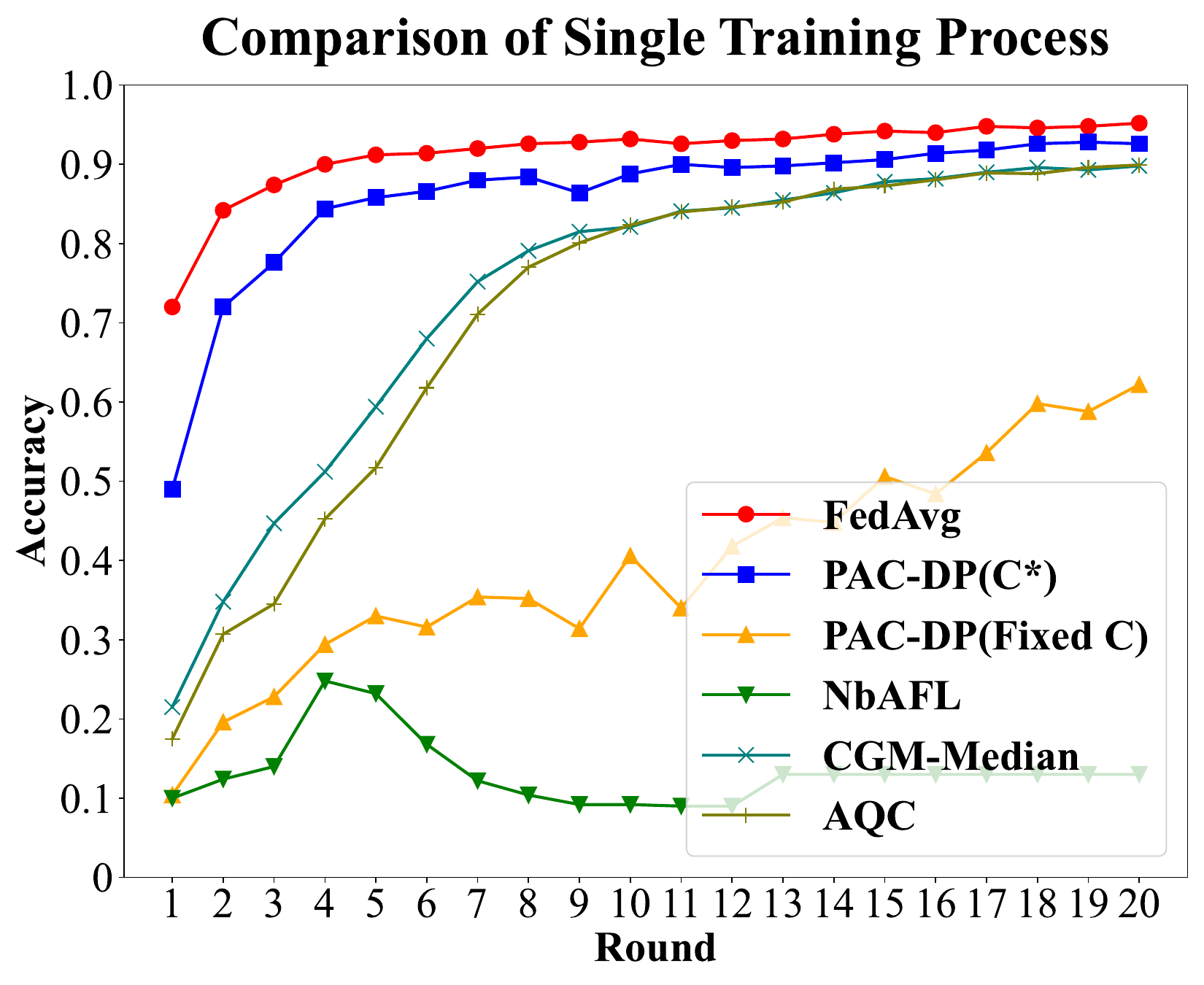}
        \caption{\small Training Dynamics \\ (Fixed Threshold)}
        \label{subfig:mnist_2}
    \end{subfigure}
    \begin{subfigure}[t]{0.23\textwidth}
        \centering
        \includegraphics[width=\textwidth]{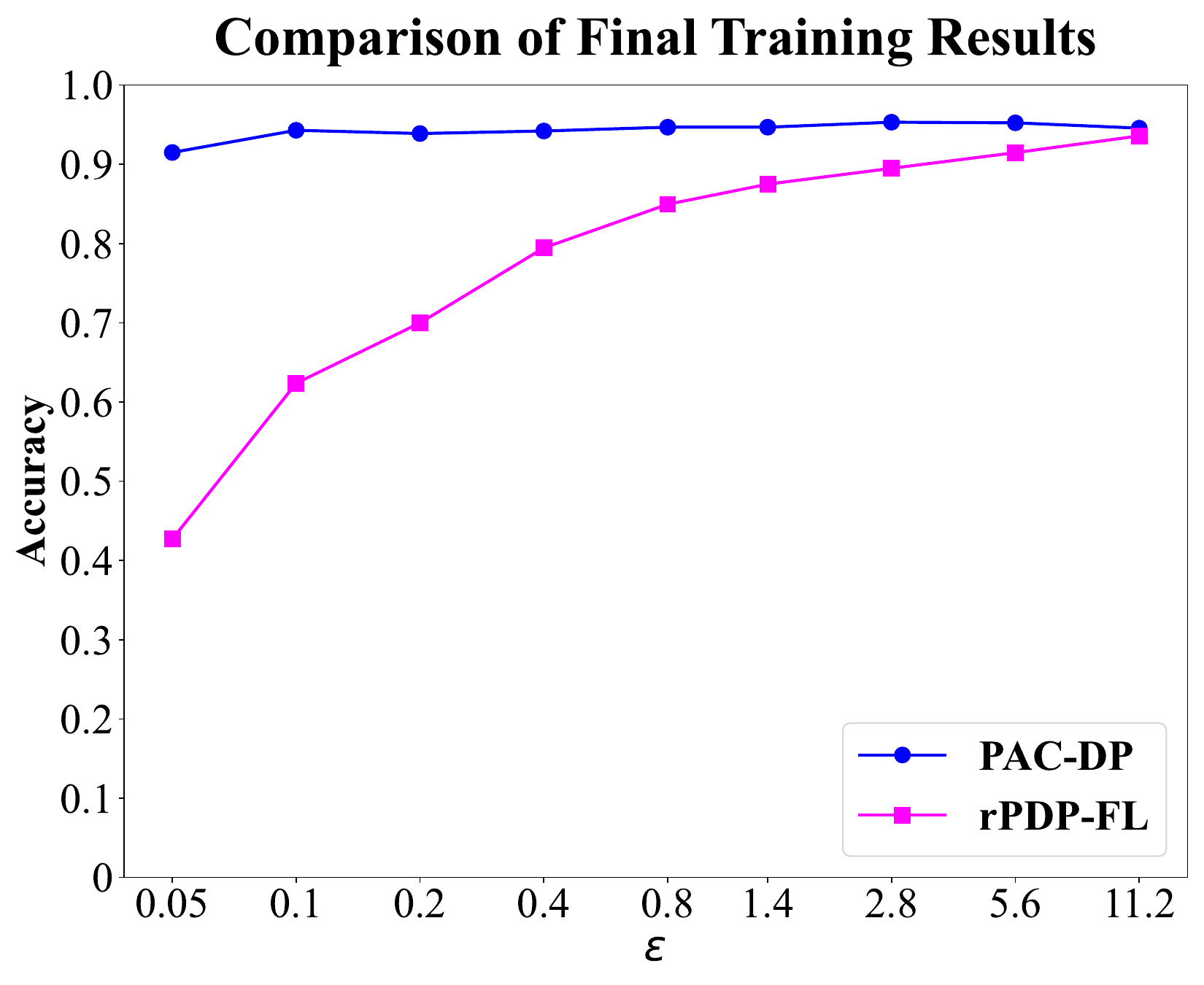}
        \caption{\small Final Accuracy \\ (rPDP-FL)}
        \label{subfig:mnist_3}
    \end{subfigure}
    \begin{subfigure}[t]{0.23\textwidth}
        \centering
        \includegraphics[width=\textwidth]{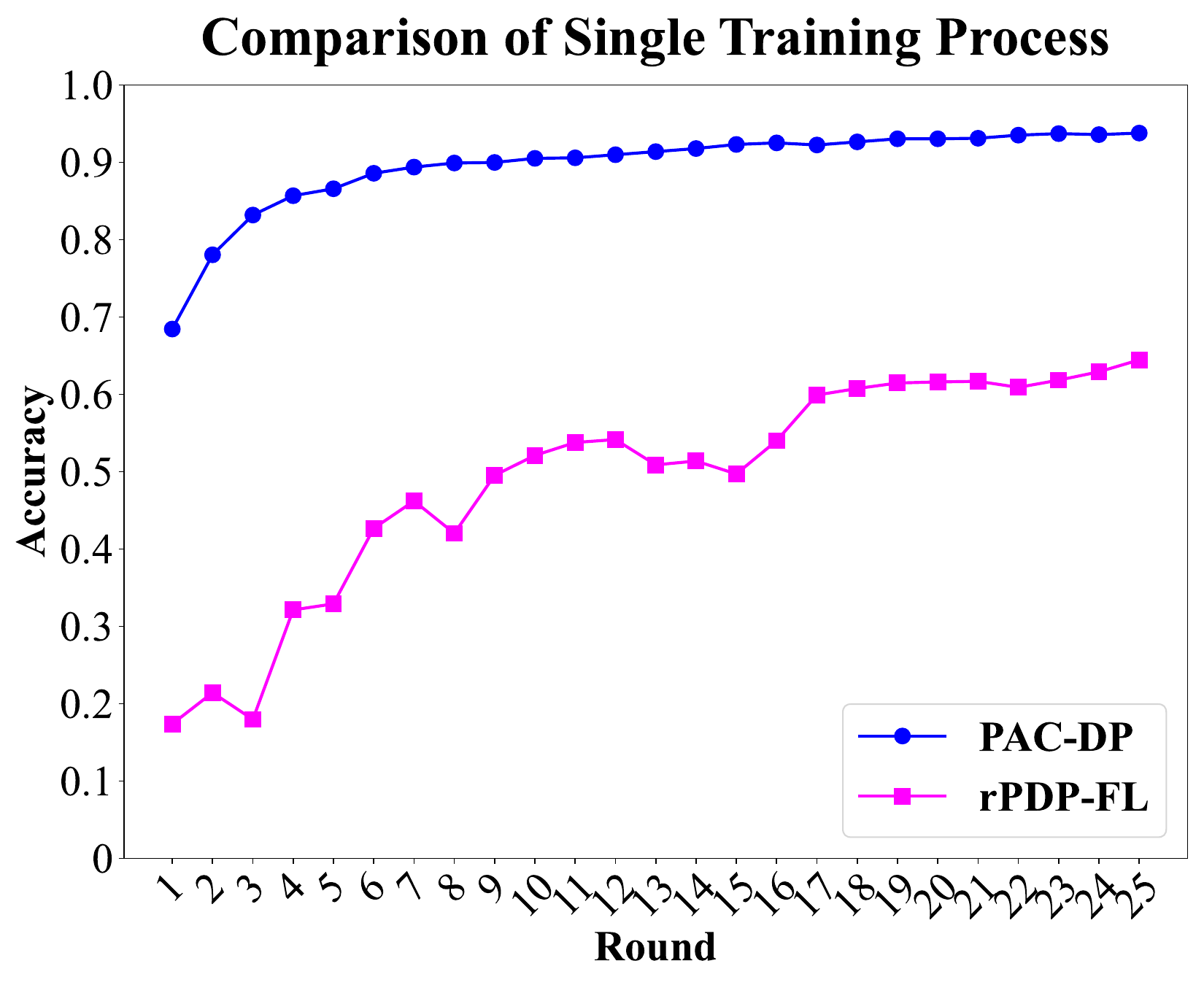}
        \caption{\small Training Dynamics \\ (rPDP-FL)}
        \label{subfig:mnist_4}
    \end{subfigure}
    \caption{Performance Comparison on MNIST under Various Privacy Budgets. PAC-DP achieves consistent accuracy gains and faster convergence across both fixed and adaptive baselines.}
    \label{fig:cmp_mnist}
\end{figure}


\textbf{Evaluation on MNIST.}
Figure~\ref{fig:cmp_mnist} compares PAC-DP with fixed-threshold and adaptive baselines (rPDP-FL, NbAFL, CGM\_Medium, and AQC) under different privacy budgets $\varepsilon$. As shown in (a) and (c), PAC-DP consistently achieves higher final accuracy, particularly under tight privacy constraints ($\varepsilon=0.1$), reaching 94.3\% versus 62.4\% (rPDP-FL), 30.5\% (NbAFL), 58.1\% (CGM\_Medium), and 49.7\% (AQC). Both CGM\_Medium and AQC set clipping thresholds based on a fixed quantile of observed gradient norms, with initial values typically chosen heuristically; this often leads to overly conservative clipping in early rounds, resulting in slower initial accuracy gains compared to PAC-DP. Training curves in (b) and (d) further illustrate that PAC-DP converges faster and more stably than all baselines, benefiting from its privacy-aware, simulation-guided threshold initialization.

\begin{figure}[htbp]
    \centering
    \begin{subfigure}[t]{0.23\textwidth}
        \centering
        \includegraphics[width=\textwidth]{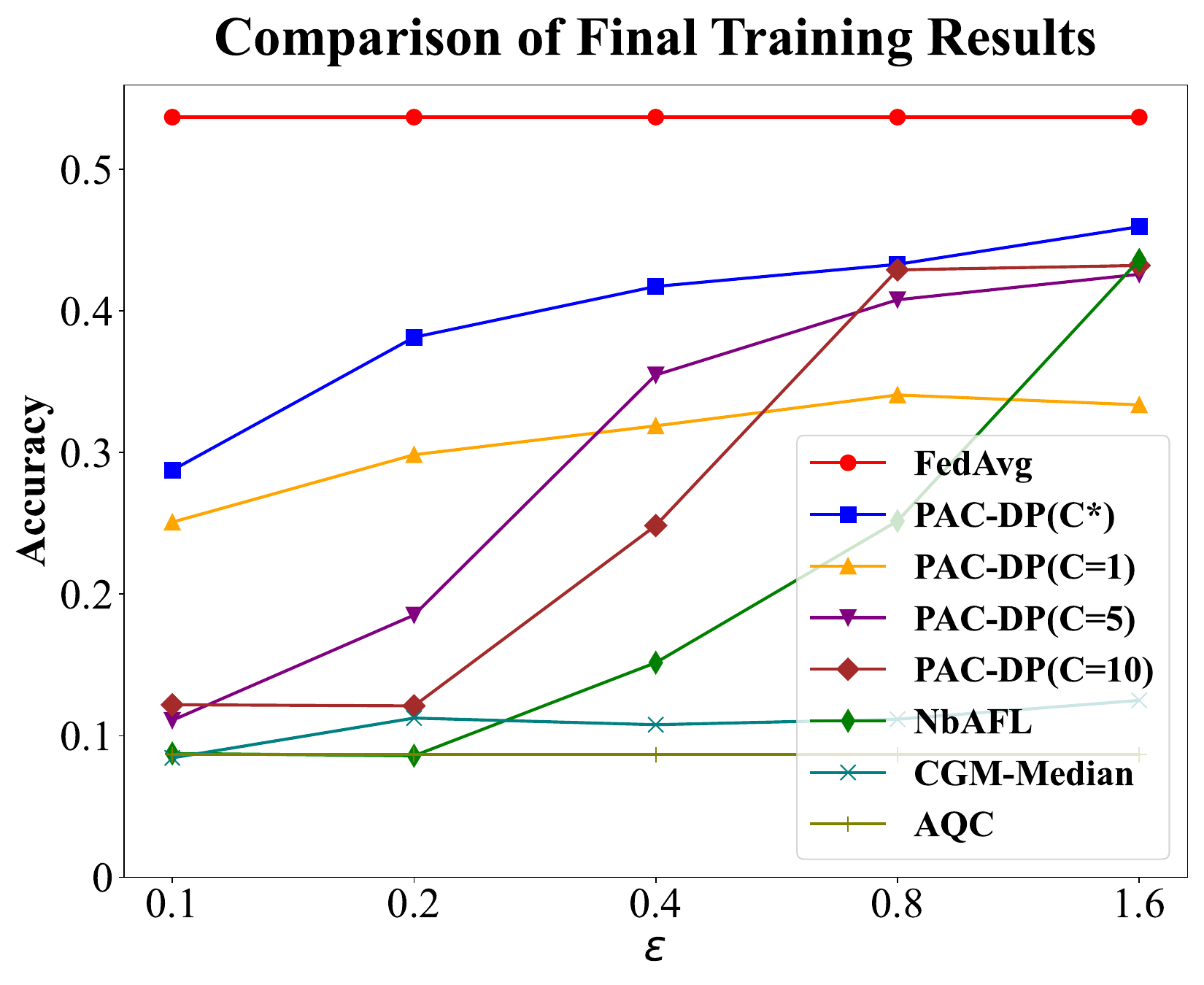}
        \caption{\small Final Accuracy \\ (Fixed Threshold)}
        \label{subfig:cifar10_1}
    \end{subfigure}
    \begin{subfigure}[t]{0.23\textwidth}
        \centering
        \includegraphics[width=\textwidth]{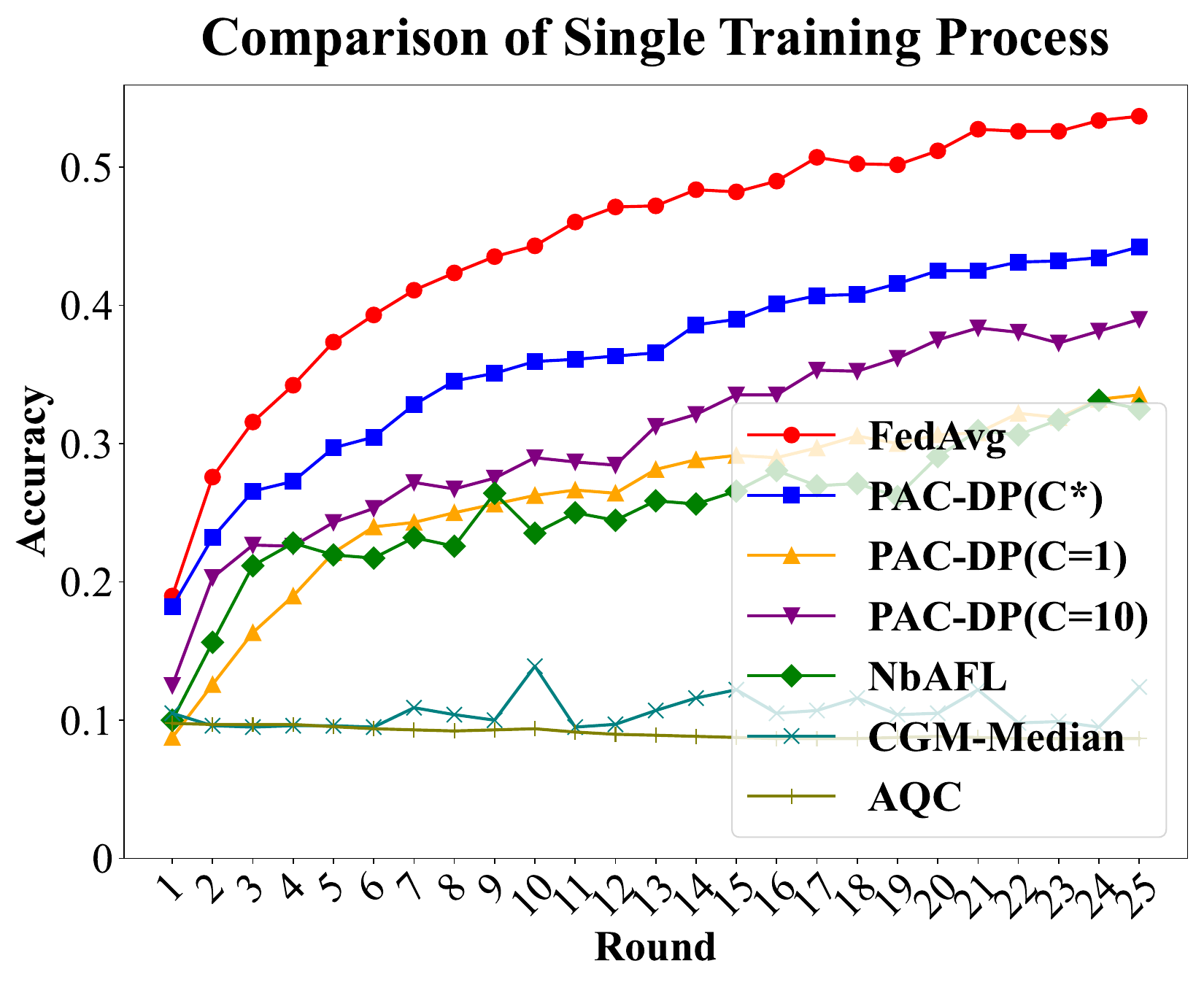}
        \caption{\small Training Dynamics \\ (Fixed Threshold)}
        \label{subfig:cifar10_2}
    \end{subfigure}
    \caption{Performance Comparison on CIFAR-10 under Fixed Thresholds and PAC-DP. Despite higher data complexity, PAC-DP maintains robust convergence and superior accuracy.}
    \label{fig:cmp_cifar10}
\end{figure}


\textbf{Evaluation on CIFAR-10.}
Figure~\ref{fig:cmp_cifar10} shows that PAC-DP maintains strong performance under more complex data. At $\varepsilon=1$, PAC-DP achieves 44.2\% accuracy, outperforming fixed threshold baselines by up to 36\%. It also exhibits faster and more stable convergence throughout training. In contrast, CGM\_Medium and AQC—while adaptive in principle—rely on empirically chosen initial clipping thresholds and gradient norm quantiles; their default configurations, effective on simpler tasks, degrade significantly on larger-scale datasets like CIFAR-10, yielding notably lower accuracy and unstable training compared to PAC-DP.

\begin{figure}[htbp]
    \centering
    \includegraphics[width=0.4\textwidth]{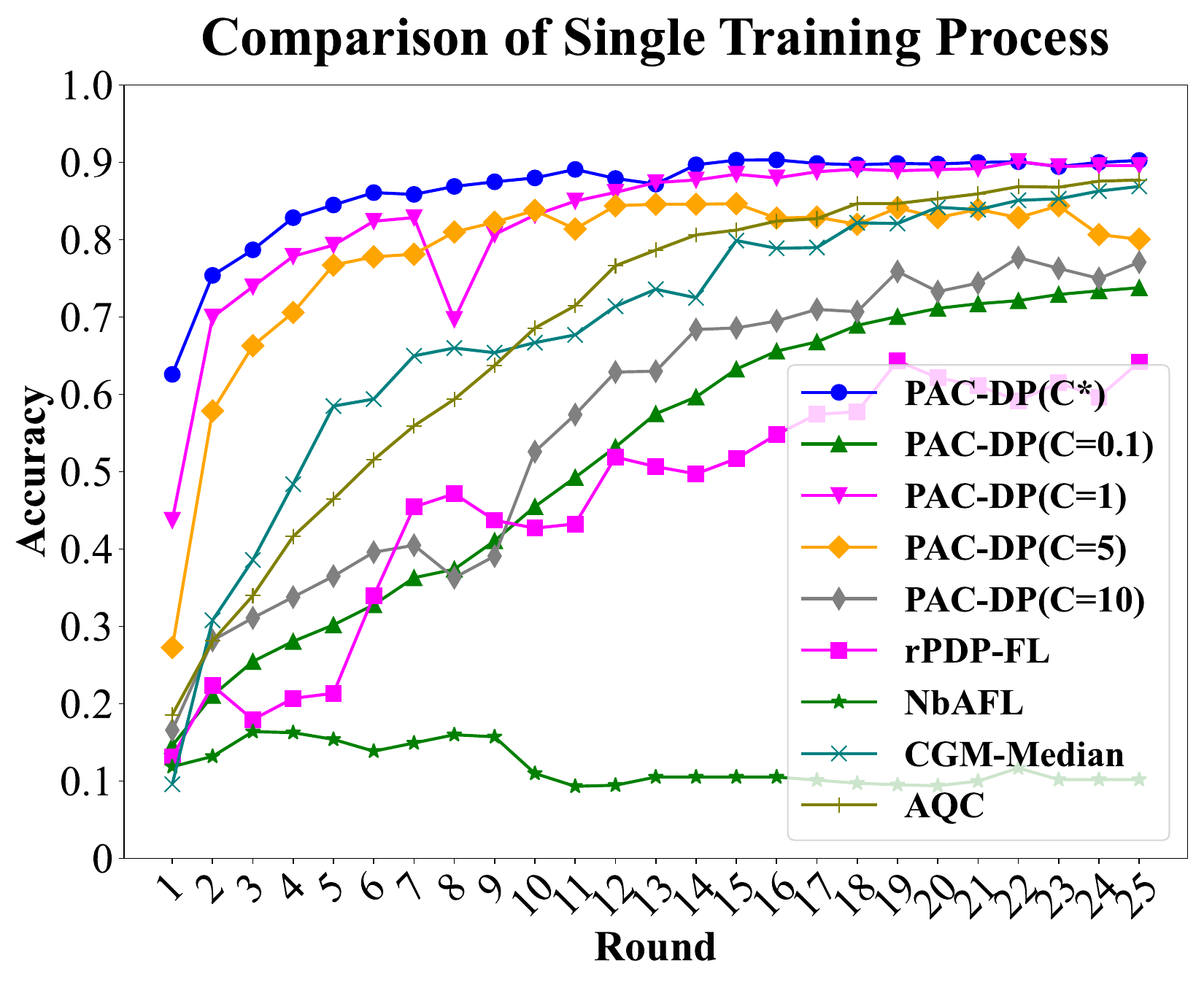}
    \caption{Training Performance under Personalized Privacy Budgets on MNIST. PAC-DP adapts effectively to heterogeneous constraints, achieving high accuracy and early convergence.}
    \label{fig:3level_mnist}
\end{figure}

\begin{table}[htbp]
\centering
\caption{Test Accuracy Comparison on MNIST (Personalized Budgets)}
\label{tab:cmp_res}
\begin{tabularx}{\columnwidth}{l|lllll}
\toprule
\textbf{Round} & 5 & 10 & 15 & 20 & 25 \\
\midrule
PAC-DP       & 0.845 & 0.880 & 0.903 & 0.898 & 0.902 \\
rPDP-FL      & 0.214 & 0.427 & 0.517 & 0.621 & 0.642 \\
NbAFL        & 0.145 & 0.110 & 0.105 & 0.094 & 0.102 \\
\end{tabularx}
\end{table}

\textbf{Heterogeneous Privacy Settings (MNIST).}
To emulate realistic user preferences, we assign privacy budgets from $\{0.05, 0.1, 1.0\}$ using sampling weights $\{0.6, 0.3, 0.1\}$. As shown in Figure~\ref{fig:3level_mnist} and Table~\ref{tab:cmp_res}, PAC-DP adapts to varying budgets and consistently outperforms both rPDP-FL and NbAFL in both early and final rounds. Notably, PAC-DP achieves over 90\% accuracy by round 5, while others remain below 65\%.

\begin{figure}[htbp]
    \centering
    \includegraphics[width=0.4\textwidth]{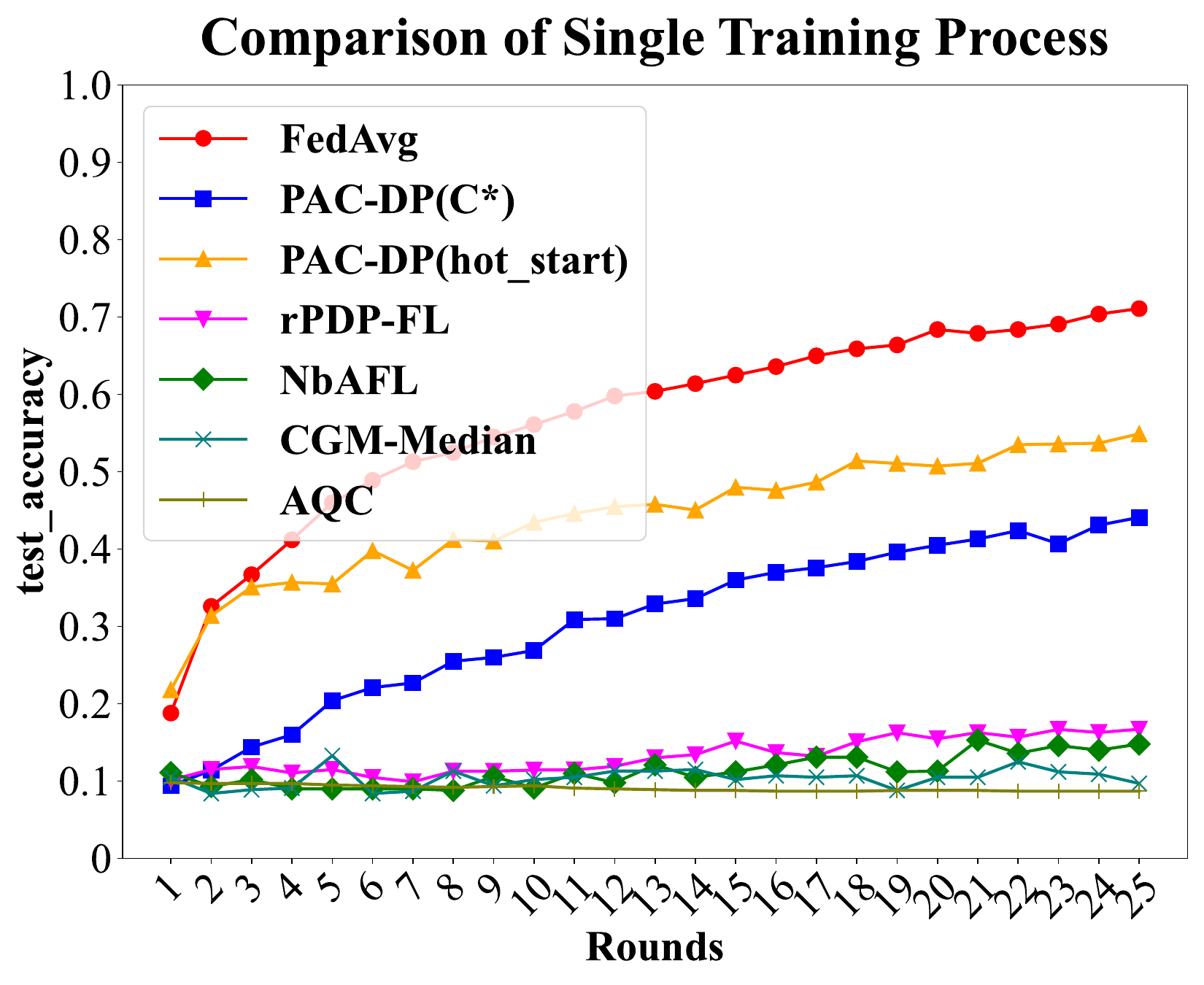}
    \caption{Testing accuracy on CIFAR-10 under Personalized Privacy Budgets. PAC-DP and PAC-DP(hot\_start) significantly outperform all baselines.}
    \label{fig:3level_adaptive_cifar10}
\end{figure}

\textbf{Heterogeneous Privacy Settings (CIFAR-10).}
We repeat the personalized budget setting on CIFAR-10 with $\varepsilon \in \{0.5, 1.0, 5.0\}$ and weights $\{0.6, 0.3, 0.1\}$. Figure~\ref{fig:3level_adaptive_cifar10} shows that PAC-DP steadily improves over rounds, reaching 44\% accuracy. Its variant PAC-DP(hot\_start) further boosts early convergence, achieving 55\% by the end. In contrast, rPDP-FL and NbAFL stagnate below 20\%.

\begin{figure}[htbp]
    \centering
    \begin{subfigure}[t]{0.23\textwidth}
        \centering
        \includegraphics[width=\textwidth]{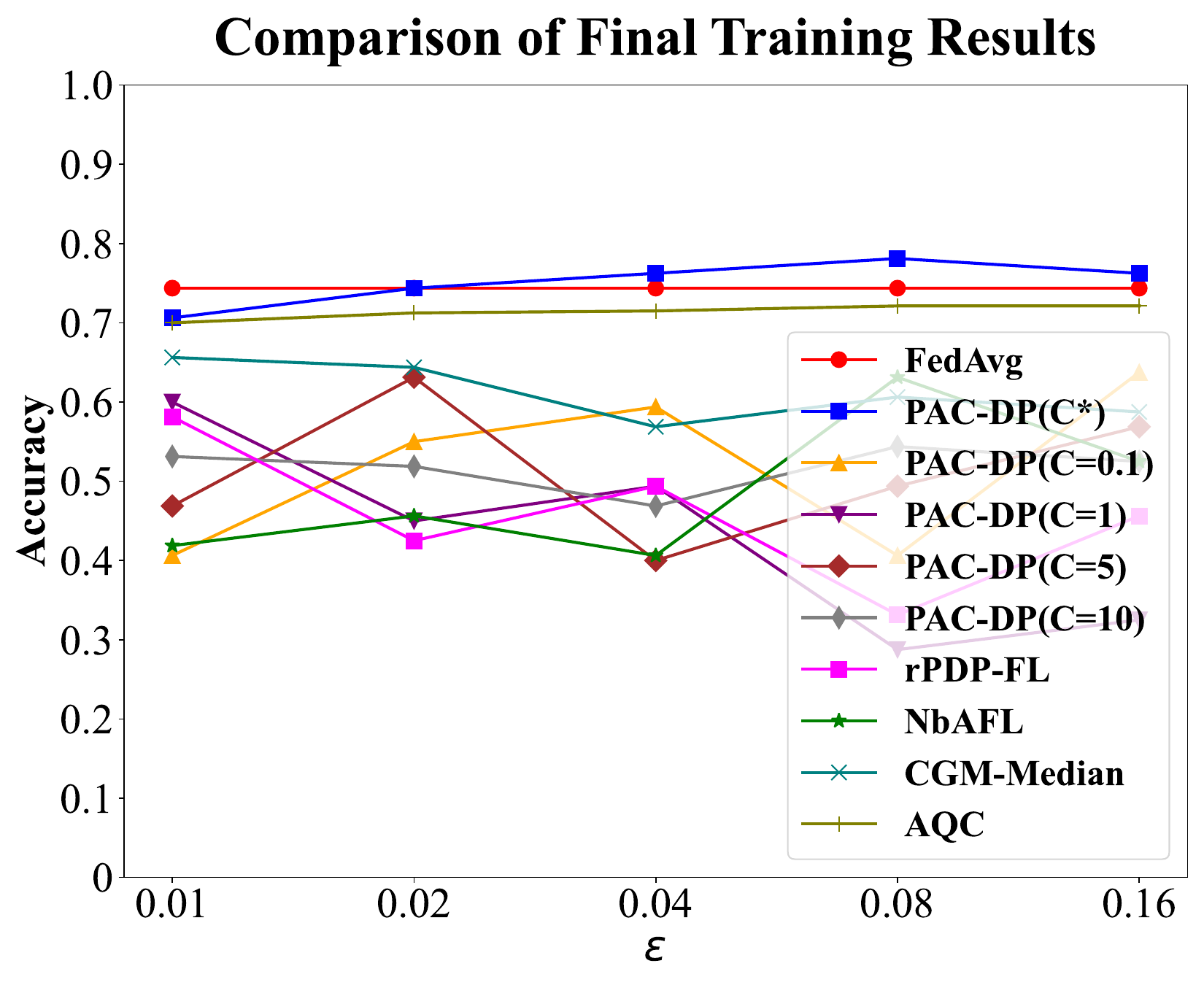}
        \caption{\small Final Accuracy \\ (Fixed Threshold)}
        \label{subfig:heart_disease_1}
    \end{subfigure}
    \begin{subfigure}[t]{0.23\textwidth}
        \centering
        \includegraphics[width=\textwidth]{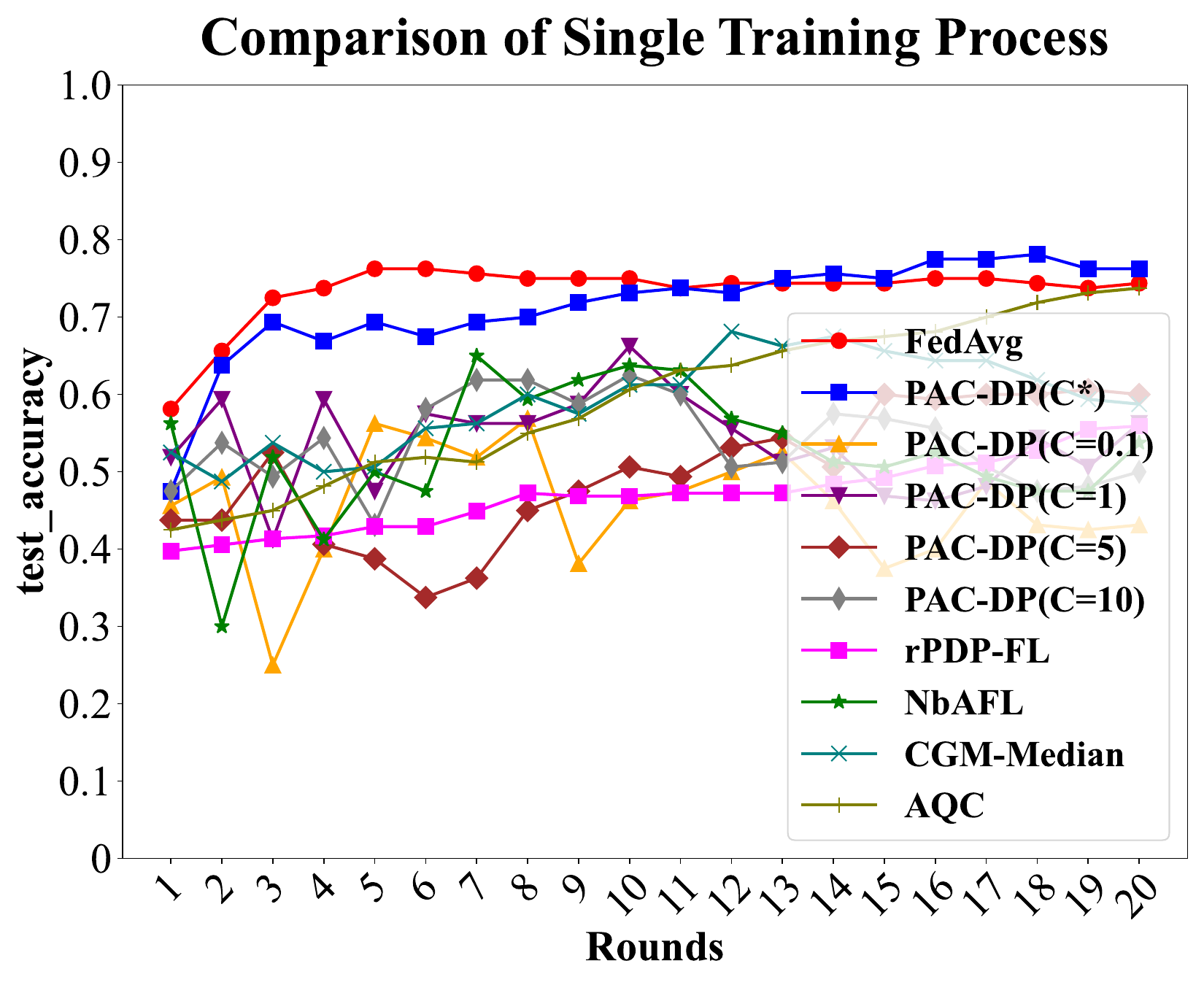}
        \caption{\small Training Dynamics \\ (Fixed Threshold)}
        \label{subfig:heart_disease_2}
    \end{subfigure}
    \caption{Performance Comparison on Fed-Heart-Disease under Fixed Thresholds and PAC-DP. Despite higher data complexity, PAC-DP maintains robust convergence and superior accuracy.}
    \label{fig:cmp_heart_disease}
\end{figure}


\textbf{Evaluation on Heart Disease.}
Figure~\ref{fig:cmp_heart_disease} shows that PAC-DP achieves robust performance on the real-world, low-dimensional tabular dataset with inherent non-IID and heterogeneous data distribution. At $\varepsilon=0.1$, PAC-DP reaches 74.4\% accuracy, outperforming all fixed threshold baselines and significantly surpassing NbAFL, CGM\_Medium, and AQC. While CGM\_Medium and AQC exhibit better performance than NbAFL and some fixed-threshold variants due to their adaptive clipping strategies, they still lag behind PAC-DP in both final accuracy and convergence stability, highlighting the advantage of our privacy-conditioned threshold design over heuristic-based quantile methods.

\begin{figure}[htbp]
    \centering
    \includegraphics[width=0.4\textwidth]{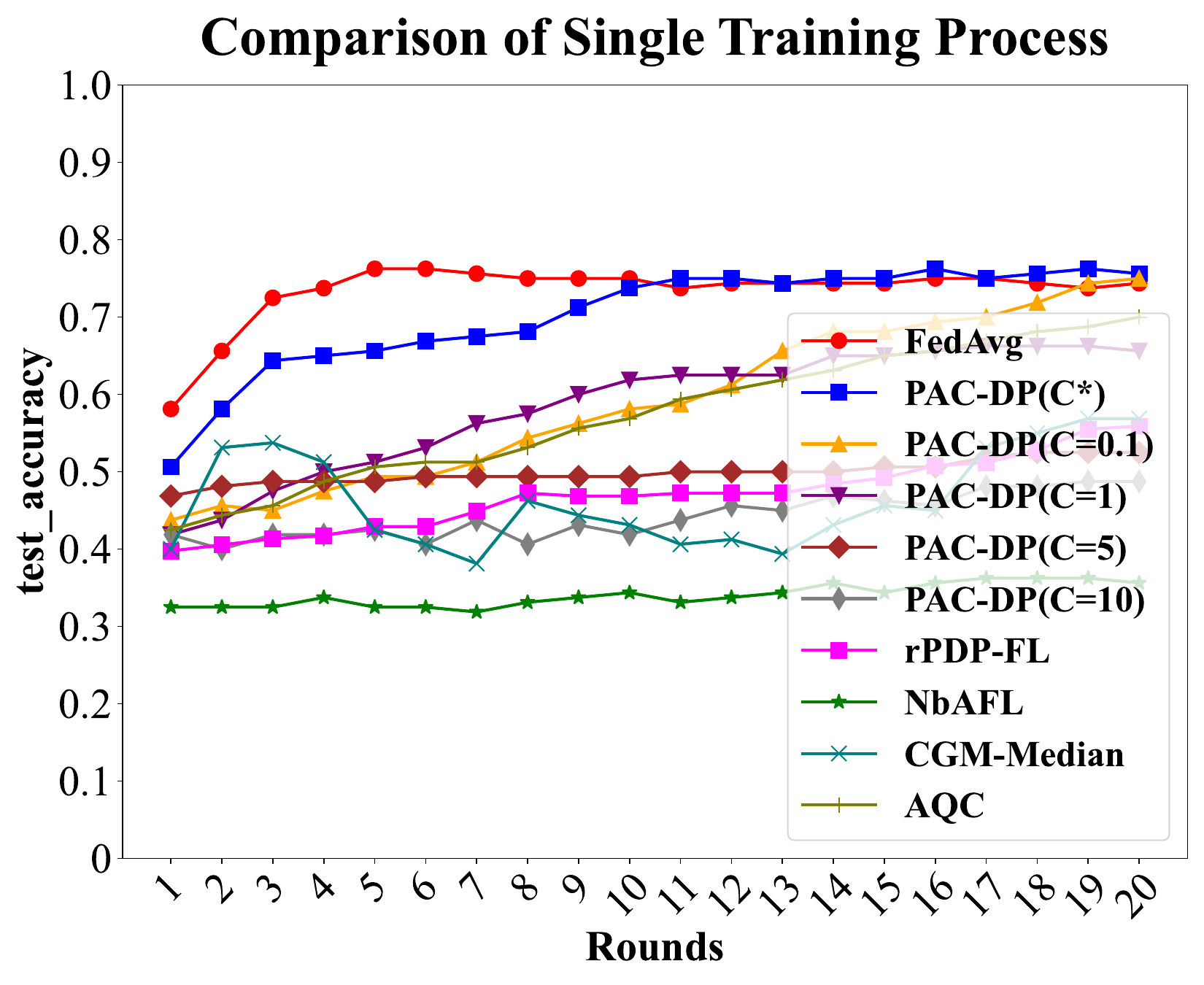}
    \caption{Testing accuracy on Fed-Heart-Disease under Personalized Privacy Budgets. PAC-DP outperform all baselines.}
    \label{fig:3level_adaptive_heart_disease}
\end{figure}

\textbf{Heterogeneous Privacy Settings (Heart Disease).}
We evaluate PAC-DP under a realistic heterogeneous privacy setting on the Fed-Heart-Disease dataset, where clients are assigned personalized privacy budgets $\varepsilon \in \{0.01, 0.05, 0.5\}$ with corresponding proportions $[0.6, 0.3, 0.1]$. Figure~\ref{fig:3level_adaptive_heart_disease} shows that PAC-DP achieves stable and rapid convergence, reaching 75.6\% test accuracy by the end of training. In contrast, fixed-threshold baselines such as PAC-DP($C=0.1$), PAC-DP($C=1$), and PAC-DP($C=10$) exhibit poor performance due to suboptimal clipping, while NbAFL stagnates below 40\%. This demonstrates that PAC-DP's adaptive clipping mechanism effectively aligns privacy constraints with heterogeneous privacy budgets and round-wise training dynamics, enabling superior utility in highly heterogeneous cross-silo environments.

\begin{figure}[htbp]
    \centering
    \includegraphics[width=0.4\textwidth]{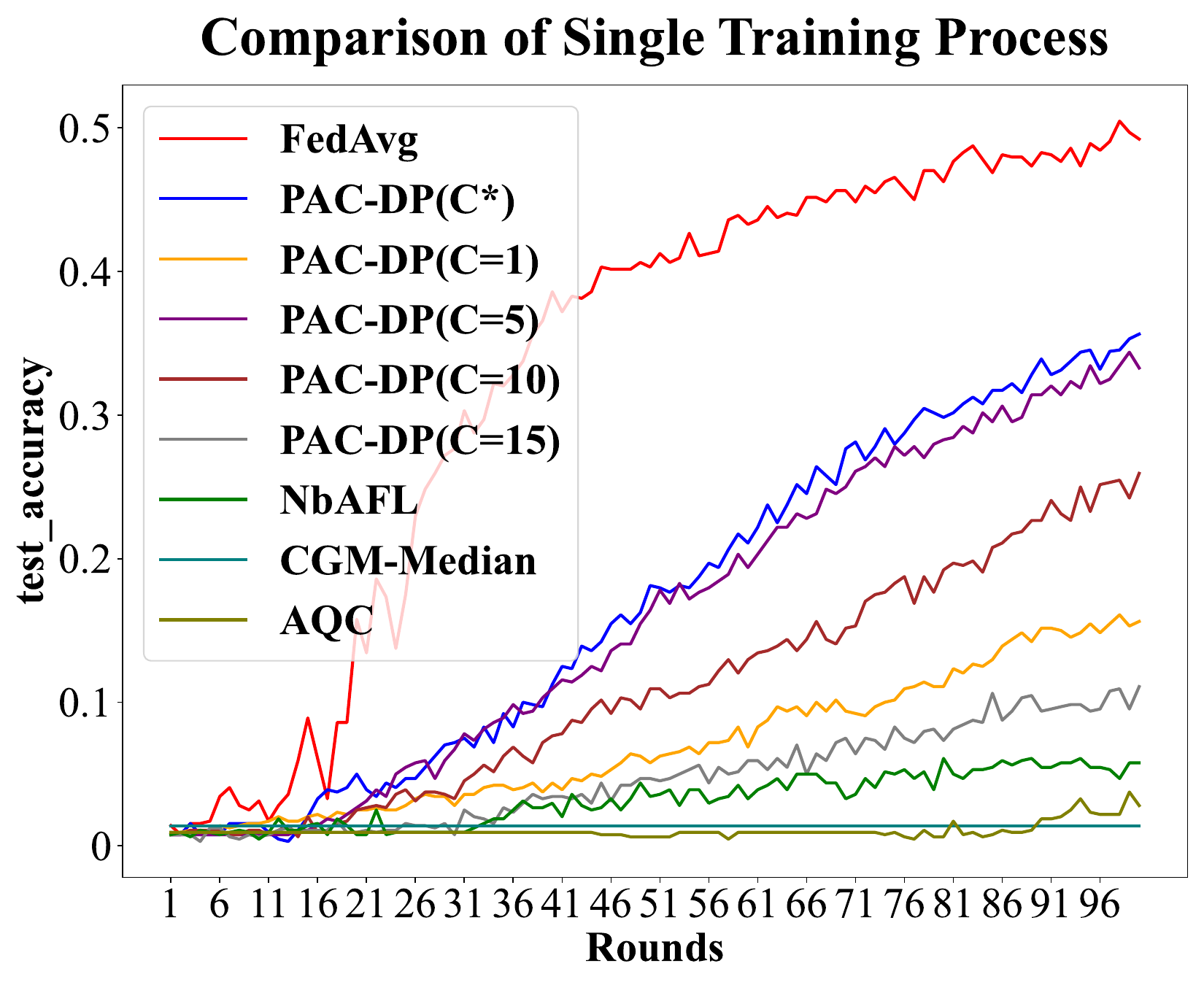}
    \caption{Testing accuracy on CIFAR-100 under Fixed Thresholds and PAC-DP.}
    \label{fig:cifar100_1}
\end{figure}


\textbf{Evaluation on CIFAR-100.}
Under a uniform privacy budget ($\varepsilon = 5.0$), PAC-DP achieves a final testing accuracy of \textbf{35.6\%}, outperforming all the fixed-threshold baselines, as shown in Figure~\ref{fig:cifar100_1}. This improvement demonstrates that adaptive clipping better balances noise injection and gradient fidelity compared to static thresholds. Similar to the observations on CIFAR-10, CGM\_Medium and AQC perform poorly on CIFAR-100, as their heuristic choices of initial thresholds and quantile levels—tuned for simpler settings—fail to generalize to this more complex and high-dimensional task.

\begin{figure}[htbp]
    \centering
    \includegraphics[width=0.4\textwidth]{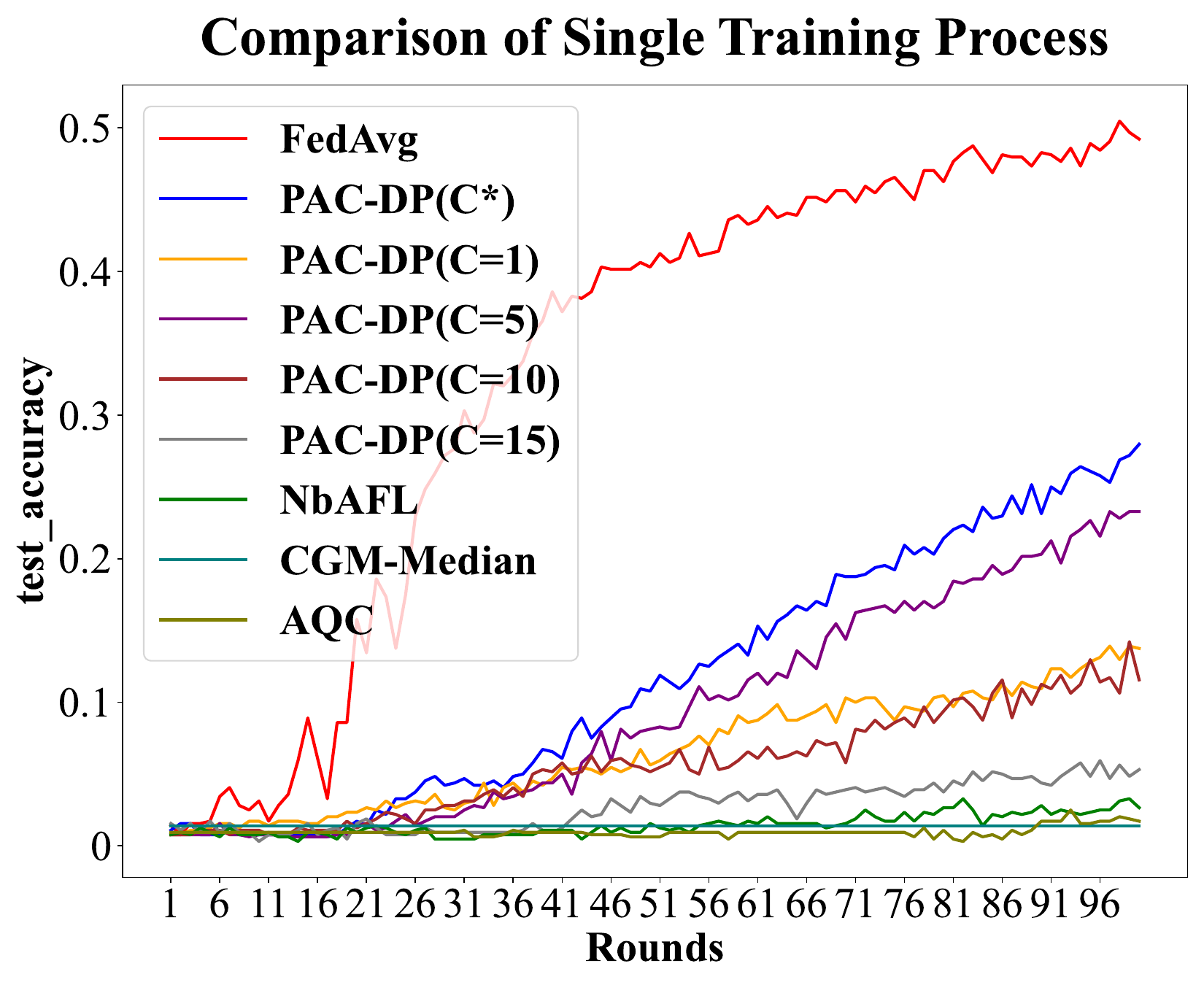}
    \caption{Testing accuracy on CIFAR-100 under Personalized Privacy Budgets. PAC-DP outperform all baselines.}
    \label{fig:cifar100_2}
\end{figure}
\textbf{Heterogeneous Privacy Settings (CIFAR-100).}
We evaluate PAC-DP under a realistic heterogeneous privacy setting on the CIFAR-100 dataset, where clients are assigned personalized privacy budgets $\varepsilon \in \{3.0, 5.0, 10.0\}$ with corresponding proportions $[0.6, 0.3, 0.1]$. PAC-DP attains a final Testing accuracy of \textbf{27.9\%}, surpassing the strongest fixed-threshold baseline (\textbf{23.2\%}) by \textbf{20.2\%}, as illustrated in Figure~\ref{fig:cifar100_2}. This confirms that tailoring clipping thresholds to individual privacy budgets effectively preserves model utility in heterogeneous environments.

\subsection{Robustness to Optimizer Choice}
\begin{figure}[htbp]
    \centering
    \includegraphics[width=0.4\textwidth]{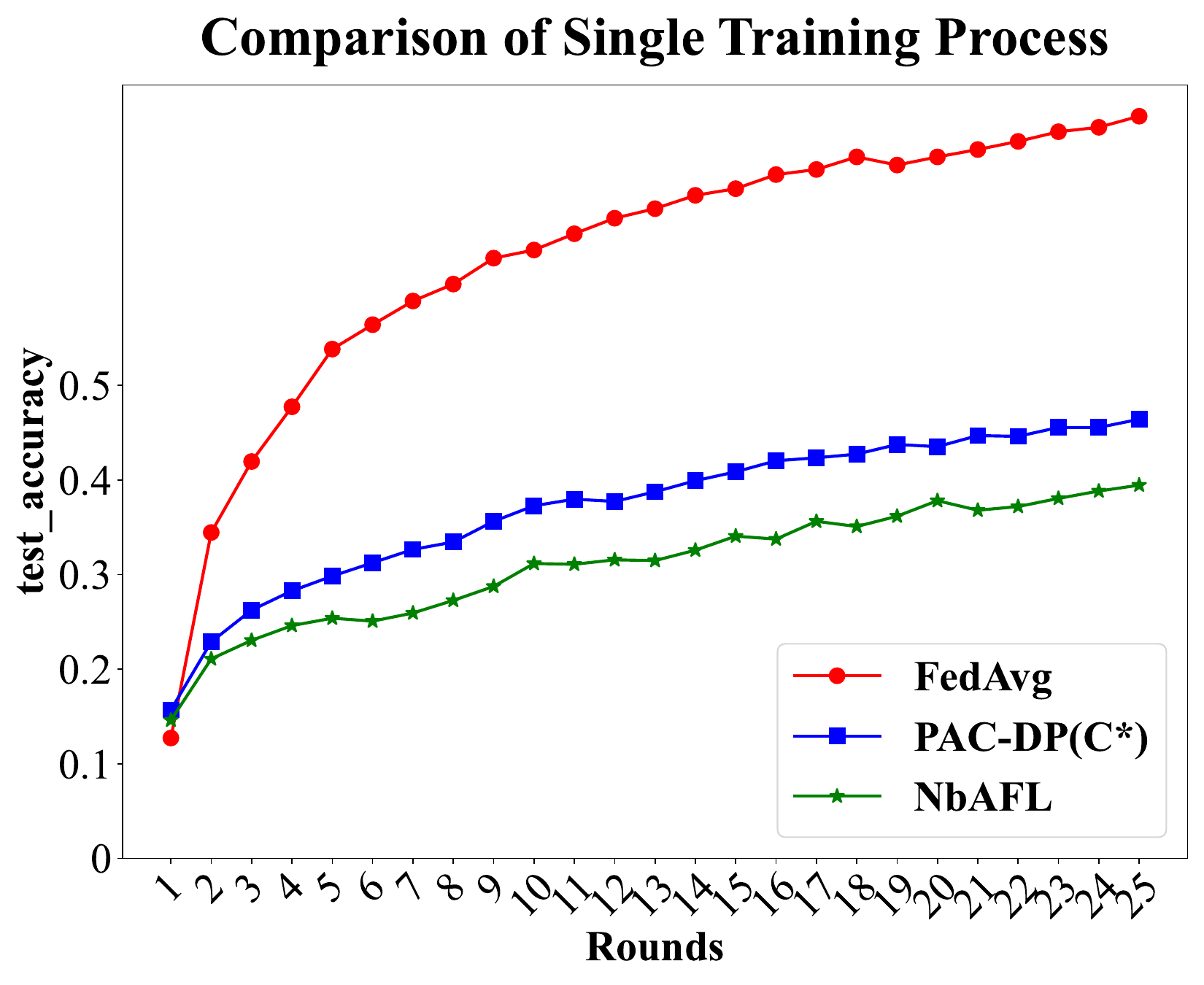}
    \caption{Testing accuracy on CIFAR-10 using the Adam Optimizer under a Fixed Clipping Threshold. PAC-DP achieve higher performance than NbAFL.}
    \label{fig:Adam_cifar10}
\end{figure}

To evaluate the robustness of PAC-DP to different optimization strategies, we replace the standard SGD optimizer with Adam—a popular adaptive optimizer known for its fast convergence and reduced sensitivity to learning rate tuning. Figure~\ref{fig:Adam_cifar10} shows the testing accuracy on CIFAR-10 when all methods use Adam with a fixed clipping threshold. Despite the change in optimizer dynamics, PAC-DP consistently outperforms NbAFL throughout training, demonstrating that our adaptive clipping mechanism remains effective under different gradient update rules.

\subsection{Communication and Computational Efficiency}

\begin{figure}[htbp]
    \centering
    \begin{subfigure}[b]{0.23\textwidth}
        \centering
        \includegraphics[width=\textwidth]{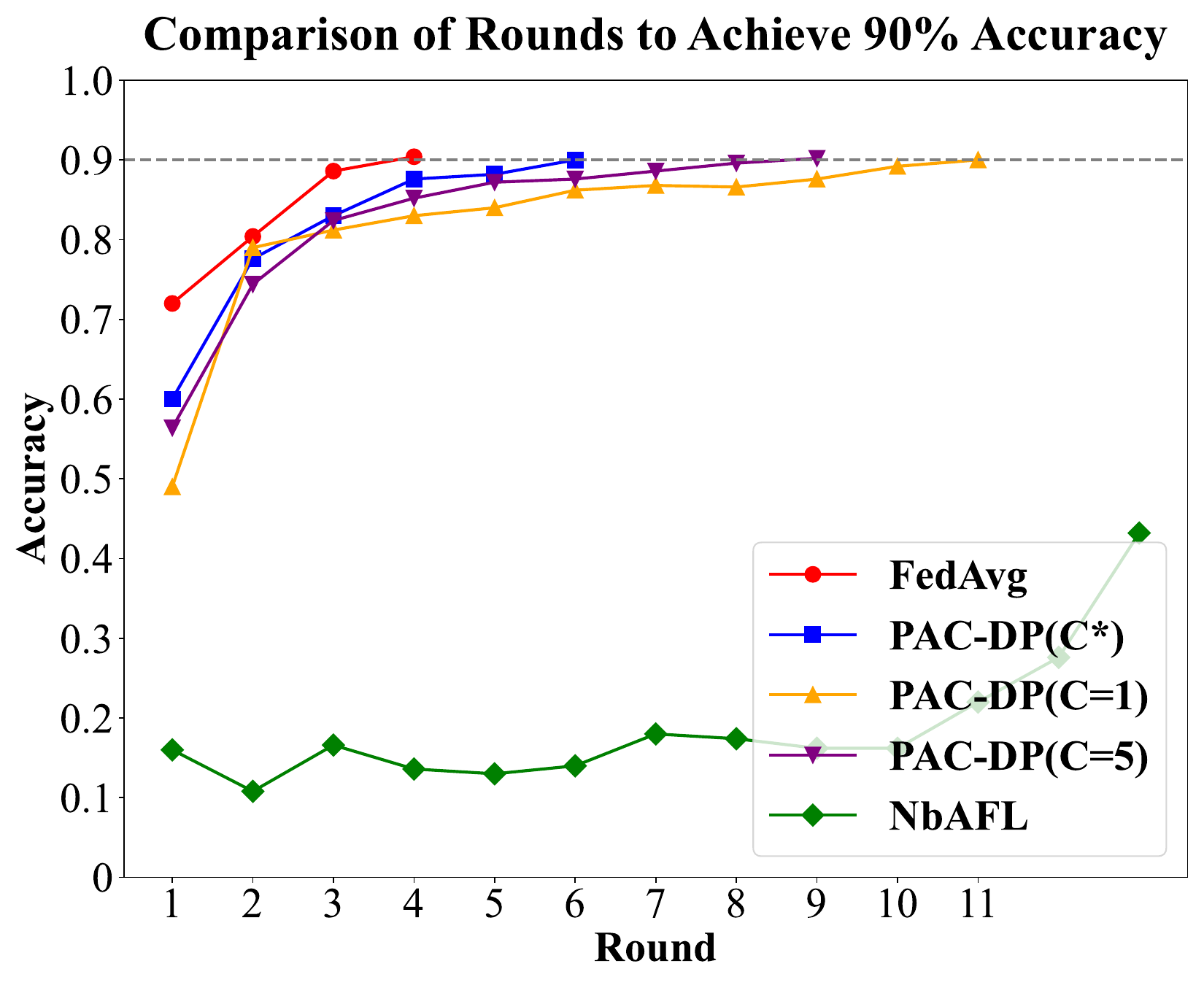}
        \caption{MNIST}
        \label{subfig:cmp_minst_round}
    \end{subfigure}
    \begin{subfigure}[b]{0.23\textwidth}
        \centering
        \includegraphics[width=\textwidth]{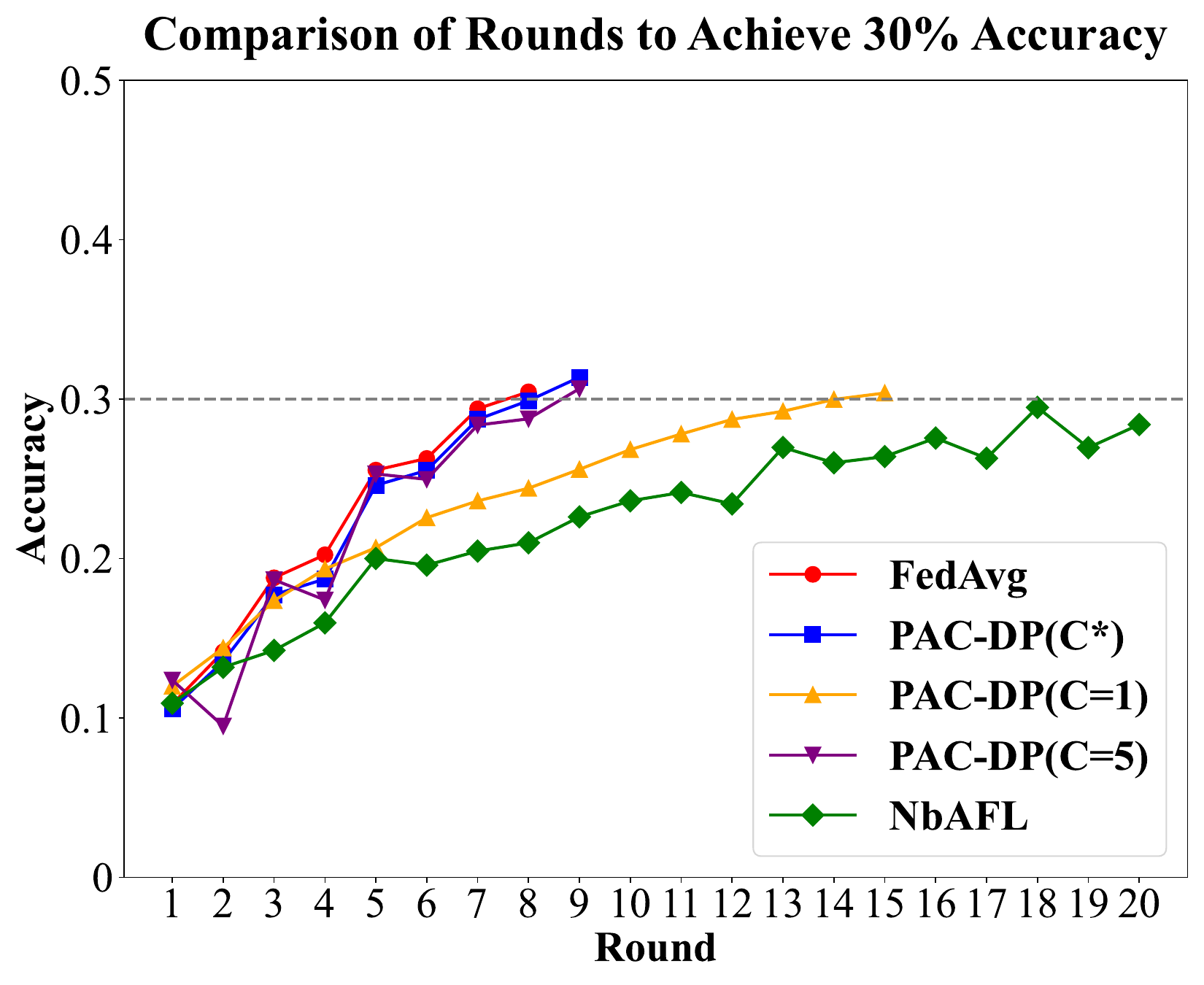}
        \caption{CIFAR-10}
        \label{subfig:cmp+cifar_round}
    \end{subfigure}
    \caption{Communication Efficiency Comparison on MNIST and CIFAR-10. PAC-DP significantly reduces the number of rounds needed to reach target accuracy.}
    \label{fig:cmp_round}
\end{figure}

To evaluate communication and convergence efficiency, we compare PAC-DP with fixed-threshold baselines and NbAFL on MNIST and CIFAR-10 under IID settings. The privacy budget is set to $\varepsilon = 0.1$.

As shown in Figure~\ref{fig:cmp_round}, PAC-DP reaches 90\% test accuracy on MNIST in only 6 rounds—saving 45.5\% and 33.3\% of communication rounds compared to $C=1$ (11 rounds) and $C=5$ (9 rounds), respectively. On CIFAR-10, PAC-DP achieves 30\% accuracy in 9 rounds, outperforming $C=1$ (15 rounds) and NbAFL (20 rounds), while matching the convergence speed of $C=5$.

The efficiency gains stem from PAC-DP’s two-phase optimization design: (1) a lightweight precomputation stage that predicts clipping thresholds based on a fitted $\varepsilon$-$C^*$ curve, avoiding manual tuning; and (2) a training stage with adaptive clipping that stabilizes convergence and mitigates oscillation caused by static thresholds.

Notably, while $C=5$ matches PAC-DP's convergence speed on CIFAR-10, it performs significantly worse on MNIST (18.2\% lower final accuracy), underscoring its lack of generalizability. In contrast, PAC-DP maintains consistently high accuracy across tasks by tailoring the clipping threshold to both privacy constraints and gradient dynamics.

\section{Conclusion and Future Work}
In this paper, we introduced PAC-DP, a Personalized Adaptive Clipping framework for federated learning with heterogeneous privacy preferences under record-level local differential privacy. 
Unlike traditional fixed-threshold methods that apply a one-size-fits-all clipping bound, PAC-DP employs a Simulation-CurveFitting approach leveraging a server-side public/synthetic proxy dataset. 
This approach learns a lightweight mapping between personalized privacy budgets $\varepsilon$ and clipping thresholds, which can be applied online with a round-wise schedule, enabling budget-conditioned calibration with minimal overhead.

Our theoretical analysis provides convergence guarantees in non-convex, convex, and strongly convex settings that are aligned with the per-example clipping and Gaussian perturbation mechanism and a reproducible privacy accounting procedure. 
These results support the role of budget-conditioned clipping in improving the privacy--utility trade-off relative to static clipping.

Comprehensive experiments on federated learning benchmarks (MNIST and CIFAR-10) empirically validate PAC-DP's advantages under matched privacy budgets. 
Across the evaluated settings, PAC-DP achieves substantial accuracy improvements (up to 26\%) and faster convergence (up to 45.5\%) compared to fixed-threshold baselines, while incurring minimal computational overhead during online training by amortizing costs through offline precomputation.
In conclusion, PAC-DP addresses an important limitation in personalized DP-FL by providing a principled and practically efficient approach to selecting clipping thresholds under heterogeneous privacy budgets.

While the results validate the benefits of budget-conditioned clipping, several avenues remain open for further exploration:
\begin{itemize}
    \item \textbf{Extension to More Complex Settings:} Extending PAC-DP to additional modalities (e.g., sequential or graph data) and larger-scale deployments with more pronounced system heterogeneity.
    \item \textbf{Dynamic Privacy Budgeting:} Integrating dynamic privacy budgeting schemes that adjust client-specific budgets over time, together with consistent privacy accounting.
    \item \textbf{Empirical Privacy Risk Evaluation:} Complementing formal DP guarantees with broader empirical privacy risk assessments under standard threat models and attack protocols.
    \item \textbf{Tighter Theoretical Bounds:} Deriving tighter convergence and privacy bounds by leveraging refined optimization analyses and tighter composition/accounting techniques.
    \item \textbf{Resource-Constrained Environments:} Optimizing PAC-DP for resource-constrained devices by balancing computation, communication, and privacy-utility performance.
\end{itemize}

\bibliographystyle{IEEEtranS}
\bibliography{IEEEabrv,myrefs}

\end{document}